%% file: ctnip_revised_new_PRSA-Main-Supplemental.tex
\documentclass[12pt,british,reprint,pra,superscriptaddress,nofootinbib,notitlepage,onecolumn]{revtex4-1}
\usepackage[T1]{fontenc}
\usepackage[latin9]{inputenc}
\usepackage{color}
\usepackage{babel}
\usepackage{amsmath}
\usepackage{amsthm}
\usepackage{amssymb}
\usepackage{graphicx}
\usepackage[unicode=true,pdfusetitle,
 bookmarks=true,bookmarksnumbered=false,bookmarksopen=false,
 breaklinks=false,pdfborder={0 0 0},pdfborderstyle={},backref=false,colorlinks=true]
 {hyperref}

\makeatletter
\numberwithin{equation}{section}
\theoremstyle{remark}
\newtheorem*{notation*}{\protect\notationname}
\theoremstyle{plain}
\newtheorem{thm}{\protect\theoremname}[section]
\theoremstyle{plain}
\newtheorem{ax}[thm]{\protect\axiomname}
\theoremstyle{plain}
\newtheorem{lem}[thm]{\protect\lemmaname}
\theoremstyle{plain}
\newtheorem{prop}[thm]{\protect\propositionname}
\theoremstyle{definition}
\newtheorem{example}[thm]{\protect\examplename}
\theoremstyle{remark}
\newtheorem{rem}[thm]{\protect\remarkname}
\theoremstyle{definition}
\newtheorem{defn}[thm]{\protect\definitionname}

\usepackage{xr}
\makeatother
\input{myQcircuit}
\makeatletter

\newcommand{\Tr}{\mathrm{Tr}}
\newcommand{\rA}{\mathrm{A}}
\newcommand{\rB}{\mathrm{B}}

\newcommand{\rR}{\mathrm{R}}
\newcommand{\rS}{\mathrm{S}}

\newcommand{\cA}{\mathcal{A}}
\newcommand{\cB}{\mathcal{B}}
\newcommand{\cC}{\mathcal{C}}

\makeatother

\providecommand{\axiomname}{Axiom}
\providecommand{\definitionname}{Definition}
\providecommand{\examplename}{Example}
\providecommand{\lemmaname}{Lemma}
\providecommand{\notationname}{Notation}
\providecommand{\propositionname}{Proposition}
\providecommand{\remarkname}{Remark}
\providecommand{\theoremname}{Theorem}

\begin{document}
\title{Necessary and Sufficient Conditions on Measurements of Quantum Channels}
\author{John Burniston}
\affiliation{Department of Mathematics \& Statistics, University of Calgary, Calgary,
AB, Canada}
\affiliation{Institute for Quantum Science and Technology, University of Calgary,
Calgary, AB, Canada}
\author{Michael Grabowecky}
\affiliation{Institute for Quantum Computing, University of Waterloo, Waterloo,
ON, Canada}
\affiliation{Department of Mathematics \& Statistics, University of Calgary, Calgary,
AB, Canada}
\affiliation{Institute for Quantum Science and Technology, University of Calgary,
Calgary, AB, Canada}
\author{Carlo Maria Scandolo}
\email{carlomaria.scandolo@ucalgary.ca}

\affiliation{Department of Mathematics \& Statistics, University of Calgary, Calgary,
AB, Canada}
\affiliation{Institute for Quantum Science and Technology, University of Calgary,
Calgary, AB, Canada}
\author{Giulio Chiribella}
\affiliation{Department of Computer Science, The University of Hong Kong, Hong
Kong, China}
\affiliation{Department of Computer Science, University of Oxford, Oxford, UK}
\author{Gilad Gour}
\affiliation{Department of Mathematics \& Statistics, University of Calgary, Calgary,
AB, Canada}
\affiliation{Institute for Quantum Science and Technology, University of Calgary,
Calgary, AB, Canada}
\begin{abstract}
Quantum supermaps are a higher-order generalization of quantum maps,
taking quantum maps to quantum maps. It is known that any completely
positive, trace non-increasing (CPTNI) map can be performed as part
of a quantum measurement. By providing an explicit counterexample
we show that, instead, not every quantum supermap sending a quantum
channel to a CPTNI map can be realized in a measurement on quantum
channels. We find that the supermaps that can be implemented in this
way are exactly those transforming quantum channels into CPTNI maps
even when tensored with the identity supermap. We link this result
to the fact that the principle of causality fails in the theory of
quantum supermaps.
\end{abstract}
\maketitle

\section{Introduction}

One of the most puzzling aspects of quantum mechanics has always been
the need to consider probabilistic processes to describe the observation
of physical systems. The development of quantum information theory
has turned this puzzling feature into a resource for many protocols.
Think, for instance, of the implementation of quantum computation
through measurements (measurement-based quantum computation) \citep{MBQC0,MBQC1},
of quantum cryptographic protocols \citep{BB84,Mayers-Yao,Device-independent1,Device-independent2},
or of the generation of random numbers \citep{QRNG}.

Focusing our attention on finite-dimensional quantum systems, the
most general quantum measuring device can be described by a set of
linear maps that are completely positive and trace non-increasing
(CPTNI). The maps in this set must sum to a quantum channel, namely
to a completely positive and trace-preserving (CPTP) linear map \citep{Nielsen2010,Wilde,Watrous}.
This situation is described by a \emph{quantum instrument} \citep{Davies1970,Holevo-book,Barchielli,Watrous}:
a quantum channel that takes a quantum system as input, and outputs
a a classical-quantum system, where the classical system represents
the `meter' read by the experimenter. From the classical outcome
read on the meter, one can infer which CPTNI map occurred during the
experiment. This characterization of quantum experiments, in conjunction
with the fact that quantum channels with trivial (i.e.\ 1-dimensional)
input represent states \citep{Chiribella2008}, singles out channels
as the fundamental objects of quantum theory, encapsulating all the
other processes.

For this reason, it is important to understand how to manipulate quantum
channels. The study of such manipulations, initiated in \citep{Chiribella2008,Hierarchy-combs,Switch},
has both practical \citep{Gisin,Gutoski,Circuit-architecture,Chiribella2008,Chiribella2016,Process-tensor,Chiribella-communication,Ebler2,Chiribella-zero,Milz_2018,Abbott,chiribella2019quantum,Chiribella-indefinite-resource}
and foundational consequences \citep{Gisin,Quartic-theory,Process-matrix,Switch,Giarmatzi,Milz_2018},
and has led to the development of new research areas, such as resource
theories of quantum processes \citep{Resource-theories,Fong,Coherence-beyond-states,Pirandola-LOCC,Berta-cost,Wilde-cost,Gour2018,Gour2018a,Coherence-beyond2,Rosset-resource,Li,Gour-review,Magic-channels,Wang_magic,Thermal-capacity,Coherence-operations,Resource-channels-1,Resource-channels-2,Gour-Winter,Gour-Scandolo,Wilde-entanglement,Process-Markov,Gaurav,Takagi-communication,Chiribella-indefinite-resource}.
The manipulation of quantum channels is implemented by \emph{supermaps}
\citep{Chiribella2008,Hierarchy-combs,Switch,Perinotti1,Perinotti2,Gour2018},
which are linear transformations sending linear maps to linear maps.
In this setting, \emph{superchannels} \citep{Chiribella2008,Gour2018}
represent the way a channel can evolve deterministically, in the same
way as channels represent the deterministic evolution of a quantum
state. Superchannels are the supermaps that take quantum channels
to quantum channels even when tensored with the identity supermap
\citep{Switch}. Measurements on channels are then described by a
set of supermaps that sum to a superchannel, giving rise to the notion
of a quantum \emph{super-instrument}.

In this article we focus on measurements performed on quantum channels,
and we show that a naive application of a condition similar to CPTNI
in quantum theory is \emph{not} enough to single out physical supermaps,
viz.\ those that can arise in an experiment performed on quantum
channels.

We will adopt the following notation.
\begin{notation*}
$\mathfrak{B}\left(\mathcal{H}\right)$ denotes the set of bounded
linear operators on the finite-dimensional Hilbert space $\mathcal{H}$,
$\mathfrak{B}_{h}\left(\mathcal{H}\right)$ the set of bounded hermitian
operators on $\mathcal{H}$, and $\mathfrak{D}\left(\mathcal{H}\right)$
the set of density matrices on $\mathcal{H}$. Every letter without
a subscript denotes a pair of systems $A:=A_{0}A_{1}$, where $A_{0}$
is usually regarded as an input system, and $A_{1}$ as an output
system. Thus $\mathcal{E}^{A}:=\mathcal{E}^{A_{0}\rightarrow A_{1}}$
denotes a linear map with input $A_{0}$ and output $A_{1}$, and
$\mathfrak{L}^{A}:=\mathfrak{L}^{A_{0}\rightarrow A_{1}}$ is the
set of such linear maps, from $\mathfrak{B}\left(\mathcal{H}^{A_{0}}\right)$
to $\mathfrak{B}\left(\mathcal{H}^{A_{1}}\right)$. $\left|A_{0}\right|$
denotes the dimension of $\mathcal{H}^{A_{0}}$. A supermap $\Theta^{A\rightarrow B}$
takes elements of $\mathfrak{L}^{A}$ to elements of $\mathfrak{L}^{B}$,
and its action on a linear map $\mathcal{E}^{A}$ is denoted with
square brackets: $\Theta^{A\rightarrow B}\left[\mathcal{E}^{A}\right]$.
Finally, a tilde over a system, as in $A_{0}\widetilde{A}_{0}$, indicates
that we are considering two identical copies of a system (in this
case $A_{0}$). We adopt the following convention concerning partial
traces: if $M^{AB}$ is a matrix on $A_{0}A_{1}B_{0}B_{1}$, $M^{AB_{0}}$
denotes the partial trace on the missing system $B_{1}$: $M^{AB_{0}}:=\mathrm{Tr}_{B_{1}}\left[M^{AB}\right]$.
In summary, when a superscript is missing, we have taken the partial
trace over the missing system of the original matrix.
\end{notation*}

\section{CPTNI-preserving supermaps}

The first condition one must require of physical supermaps is that
they be \emph{completely CP-preserving} (CPP): they should send CP
maps to CP maps even when tensored with the identity supermap. In
formula, a supermap $\Theta^{A\rightarrow B}$ is CPP if for all bipartite
CP maps $\mathcal{E}^{RA}\in\mathfrak{L}^{RA}$, we have that
\begin{equation}
\left(\mathbf{1}^{R}\otimes\Theta^{A\rightarrow B}\right)\left[\mathcal{E}^{RA}\right],\label{eq:CPP}
\end{equation}
is still a CP map, where $\mathbf{1}^{R}:=\mathbf{1}^{R\rightarrow R}$
is the identity supermap. This is analogous to the CP condition for
quantum maps.

The second condition, analogous to being TNI for quantum maps, is
that a physical supermap should send CPTNI maps to CPTNI maps. If
a supermap is CPP, demanding this is equivalent to requiring that
it should take \emph{CPTP} maps to CPTNI maps (see appendix~\ref{sec:General-facts-about}).
More precisely, a supermap $\Theta^{A\rightarrow B}$ is \emph{CPTNI-preserving}
if it is CPP and
\begin{equation}
\mathrm{Tr}\left[\Theta^{A\rightarrow B}\left[\mathcal{N}^{A}\right]\left(\rho^{B_{0}}\right)\right]\leq1,\label{eq:CPTNI}
\end{equation}
for any CPTP map $\mathcal{N}^{A}\in\mathfrak{L}^{A}$ and any $\rho^{B_{0}}\in\mathfrak{D}\left(\mathcal{H}^{B_{0}}\right)$.
The analogy between CPTNI quantum maps and CPTNI-preserving supermaps
is illustrated in Fig.~\ref{fig:CPTNI-CPTNI}.
\begin{figure}
\begin{centering}
\includegraphics[scale=0.3]{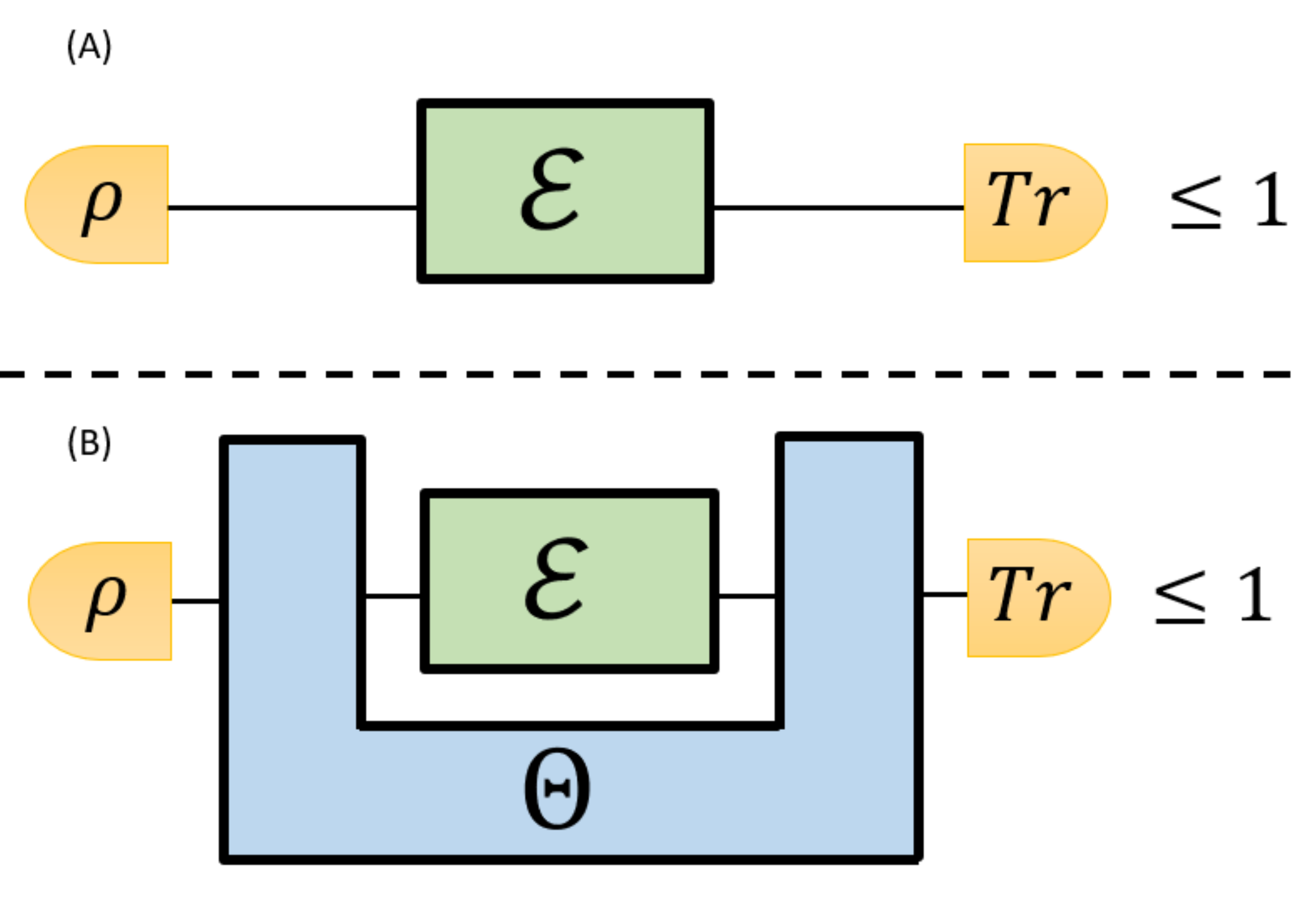}
\par\end{centering}
\caption{\label{fig:CPTNI-CPTNI}The condition for a CP map $\mathcal{E}$
to be a CPTNI map is represented in (a). The condition for a CPP supermap
$\Theta$ to be a CPTNI-preserving supermap is depicted in (b). Note
that the two conditions are formally very similar, the only difference
being the presence of the vice representing the supermap $\Theta$.
Therefore, we can rightfully say that CPTNI-preserving supermaps are
the higher-level analogy of CPTNI maps.}

\end{figure}

A measurement on quantum channels (called a \emph{super-measurement})
is described by a set of CPTNI-preserving supermaps $\left\{ \Theta_{x}^{A\rightarrow B}\right\} _{x\in X}$,
indexed by the outcome $x$ of the measurement, such that $\sum_{x\in X}\Theta_{x}^{A\rightarrow B}$
is a superchannel. This gives rise to the super-instrument:
\begin{equation}
\Upsilon^{A\rightarrow X_{1}B}\left[\mathcal{E}^{A}\right]=\sum_{x\in X}\left|x\right\rangle \left\langle x\right|^{X_{1}}\otimes\Theta_{x}^{A\rightarrow B}\left[\mathcal{E}^{A}\right],\label{eq:super-instrument}
\end{equation}
for every CP map $\mathcal{E}^{A}$, where system $X_{1}$ represents
the classical meter and $\left\{ \left|x\right\rangle ^{X_{1}}\right\} $
is an orthonormal basis of $X_{1}$.

Our main result is that, surprisingly, \emph{not all} CPTNI-preserving
supermaps can arise in a quantum super-measurement, therefore \emph{not
all} CPTNI-preserving supermaps are physical. An example is the supermap
$\Theta^{A\rightarrow B}$ whose action on a generic CP map $\mathcal{E}^{A}$
is:
\begin{equation}
\Theta^{A\rightarrow B}\left[\mathcal{E}^{A}\right]\left(\rho^{B_{0}}\right)=\mathrm{Tr}\left[\mathcal{E}^{A_{0}\rightarrow B_{0}}\left(u^{A_{0}}\right)Y^{B_{0}}\left(\rho^{B_{0}}\right)^{\mathrm{T}}Y^{B_{0}}\right]u^{B_{1}},\label{eq:example}
\end{equation}
where all systems are qubits, $u$ is the maximally mixed state, and
$Y$ is the Pauli $Y$ matrix ($\rho^{B_{0}}$ is a generic density
matrix, used to define the action of the CPTNI map $\Theta^{A\rightarrow B}\left[\mathcal{E}^{A}\right]$
on its input). Note that, if $\mathcal{E}^{A}$ is CPTP, one has indeed
$\mathrm{Tr}\left[\Theta^{A\rightarrow B}\left[\mathcal{E}^{A}\right]\left(\rho^{B_{0}}\right)\right]\leq1$,
because $\rho^{B_{0}}$ is a density matrix. This means that the supermap
$\Theta^{A\rightarrow B}$ is CPTNI-preserving (cf.\ Eq.~\eqref{eq:CPTNI}).
Full details are presented in appendix~\ref{sec:counterexample}.

\section{Completely CPTNI-preserving supermaps}

We find that the right condition to ensure that a CPTNI-preserving
supermap $\Theta^{A\rightarrow B}$ is physical is that it be \emph{completely}
CPTNI-preserving. This means that it is CPTNI-preserving even when
tensored with the identity supermap:
\begin{equation}
\mathrm{Tr}\left[\left(\mathbf{1}^{R}\otimes\Theta^{A\rightarrow B}\right)\left[\mathcal{N}^{RA}\right]\left(\rho^{R_{0}B_{0}}\right)\right]\leq1,\label{eq:c-CPTNI}
\end{equation}
where $\mathcal{N}^{RA}$ is a CPTP map, and $\rho^{R_{0}B_{0}}\in\mathfrak{D}\left(\mathcal{H}^{R_{0}B_{0}}\right)$.
The example in Eq.~\eqref{eq:example} highlights that, in general,
not all CPTNI-preserving supermaps are completely CPTNI-preserving.

For superchannels the situation is different: it is sufficient to
demand that they be  CPP and TP-preserving (TPP), without requiring
that they be TPP in a \emph{complete} sense \citep{Gour2018}. The
situation of generic supermaps is also different from linear maps
acting on quantum states. In the latter case, to have a physical CP
map, it is enough to require that it be TNI, without demanding it
in a complete sense. The ultimate reason for these different behaviours
is related to causality and no-signalling \citep{Chiribella-purification},
and it is fully examined in appendix~\ref{sec:OPT-interpretation-of}.

\section{The main result}

Following \citep{Hierarchy-combs,Gour2018}, we work in the Choi picture
for quantum maps and supermaps. A summary of useful facts is presented
in appendix~\ref{subsec:Useful-results-about}. Let $\mathbf{J}_{\Theta}^{AB}$
be the Choi matrix of a supermap $\Theta^{A\rightarrow B}$, and $J_{\mathcal{E}}^{A}$
the Choi matrix of a linear map $\mathcal{E}^{A}\in\mathfrak{L}^{A}$.
Then $\Theta^{A\rightarrow B}$ is a CPTNI-preserving supermap if
and only if $\mathbf{J}_{\Theta}^{AB}\geq0$ (since it is  CPP), and
it satisfies the additional condition deriving from Eq.~\eqref{eq:CPTNI}:
\begin{equation}
\mathrm{Tr}\left[\mathbf{J}_{\Theta}^{AB_{0}}\left(J_{\mathcal{N}}^{A}\otimes\rho^{B_{0}}\right)^{\mathrm{T}}\right]\leq1,\label{eq:choicond}
\end{equation}
for every CPTP map $\mathcal{N}^{A}\in\mathfrak{L}^{A}$, and every
$\rho^{B_{0}}\in\mathfrak{D}\left(\mathcal{H}^{B_{0}}\right)$ (see
appendix~\ref{subsec:Useful-results-about}). Notice the similarity
with the definition of CPTNI maps $\mathcal{E}$ in the Choi picture,
namely
\[
\mathrm{Tr}\left[J_{\mathcal{E}}^{A}\left(\left(\rho^{A_{0}}\right)^{\mathrm{T}}\otimes I^{A_{1}}\right)\right]\leq1,
\]
for every $\rho^{A_{0}}\in\mathfrak{D}\left(\mathcal{H}^{A_{0}}\right)$.

In a similar spirit, we can express the requirement of complete CPTNI
preservation in Eq.~\eqref{eq:c-CPTNI} in the Choi picture as $\mathbf{J}_{\Theta}^{AB}\geq0$
plus the remarkably simple additional constraint
\begin{equation}
\mathrm{Tr}\left[\mathbf{J}_{\Theta}^{AB_{0}}\left(M^{AB_{0}}\right)^{\mathrm{T}}\right]\leq1,\label{eq:completechoicond}
\end{equation}
for every positive semi-definite matrix $M^{AB_{0}}$ with marginal
$M^{A_{0}B_{0}}=I^{A_{0}}\otimes\rho^{B_{0}}$, for some $\rho^{B_{0}}\in\mathfrak{D}\left(\mathcal{H}^{B_{0}}\right)$.
The technical details are provided in appendix~\ref{subsec:Some-technical-derivations}.

It is not hard to check that all the matrices of the form $J_{\mathcal{N}}^{A}\otimes\rho^{B_{0}}$,
with $\mathcal{N}^{A}$ CPTP, are a strict subset of the matrices
$M^{AB_{0}}$, confirming that complete CPTNI preservation is at least
as strict a condition as CPTNI preservation. In fact, it is stricter,
as our counterexample in Eq.~\eqref{eq:example} shows: the supermap
in Eq.~\eqref{eq:example} is CPTNI-preserving but \emph{not completely}
CPTNI-preserving. Consequently, the set of completely CPTNI-preserving
supermaps is strictly contained in the set of CPTNI-preserving supermaps.
The situation is illustrated in Fig.~\ref{fig:Inclusions}.
\begin{figure}
\begin{centering}
\includegraphics[scale=0.9]{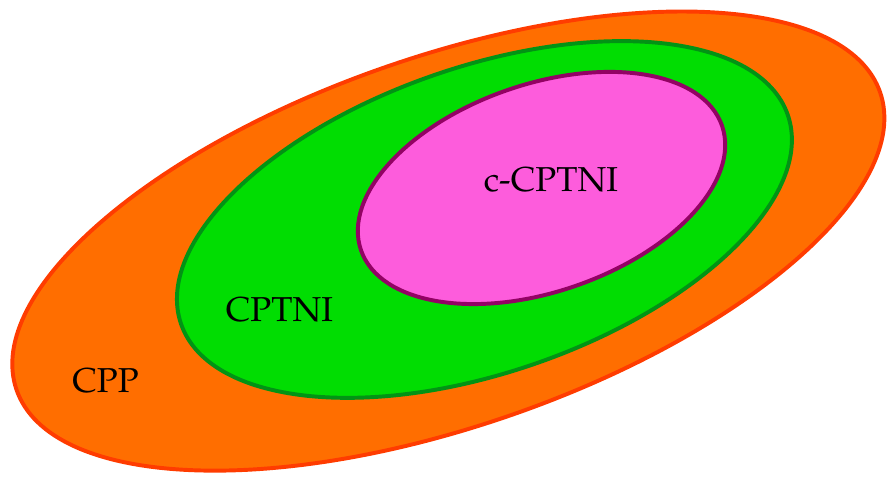}
\par\end{centering}
\caption{\label{fig:Inclusions}Inclusions between sets of supermaps. Here
c-CPTNI denotes completely CPTNI-preserving supermaps.}
\end{figure}

To obtain our main result, namely the characterization of which CPTNI-preserving
supermaps are physical, we consider a semi-definite program (SDP)
inspired by Eq.~\eqref{eq:completechoicond}:
\begin{eqnarray}
\textrm{Find} & \quad & \alpha=\max_{M}\mathrm{Tr}\left[\mathbf{J}_{\Theta}^{AB_{0}}\left(M^{AB_{0}}\right)^{\mathrm{T}}\right]\nonumber \\
\textrm{Subject to:} & \quad & M^{AB_{0}}\geq0\nonumber \\
 & \quad & M^{A_{0}B_{0}}=I^{A_{0}}\otimes\rho^{B_{0}}.\label{eq:completeSDPprimal}
\end{eqnarray}
If we consider the dual of the SDP~\eqref{eq:completeSDPprimal}
\begin{eqnarray*}
\textrm{Find} & \quad & \beta=\left|A_{0}\right|\min r\\
\textrm{Subject to:} & \quad & r\left|A_{0}\right|\mathbf{J}_{\Phi}^{A_{0}B_{0}}\otimes u^{A_{1}}\geq\mathbf{J}_{\Theta}^{AB_{0}}\\
 & \quad & \mathbf{J}_{\Phi}^{A_{0}B_{0}}\geq0\\
 & \quad & \mathbf{J}_{\Phi}^{A_{1}B_{0}}=I^{A_{1}B_{0}}\\
 & \quad & r\geq0\\
 & \quad & r\in\mathbb{R}\\
 & \quad & \Phi\textrm{ superchannel},
\end{eqnarray*}
we convert Eq.~\eqref{eq:completeSDPprimal} from an SDP having a
constraint on $M^{AB_{0}}$ into one having an explicit condition
on $\mathbf{J}_{\Theta}^{AB_{0}}$. This condition is exactly what
we need to derive our main result.
\begin{thm}
A CPTNI-preserving supermap can be completed to a superchannel if
and only if it is \emph{completely} CPTNI-preserving.
\end{thm}

The full proof is presented in appendix~\ref{sec:Completion-of-Supermaps}.

\section{Realization of completely CPTNI-preserving supermaps}

Using the Choi picture, we can re-obtain a result of \citep[theorem 2]{Chiribella2008},
namely that every completely CPTNI-preserving supermap can be expressed
in terms of a CPTP pre-processing map and a CPTNI post-processing
map, as depicted in Fig.~\ref{fig:Completely}.
\begin{figure}
\begin{centering}
\includegraphics[width=0.5\columnwidth]{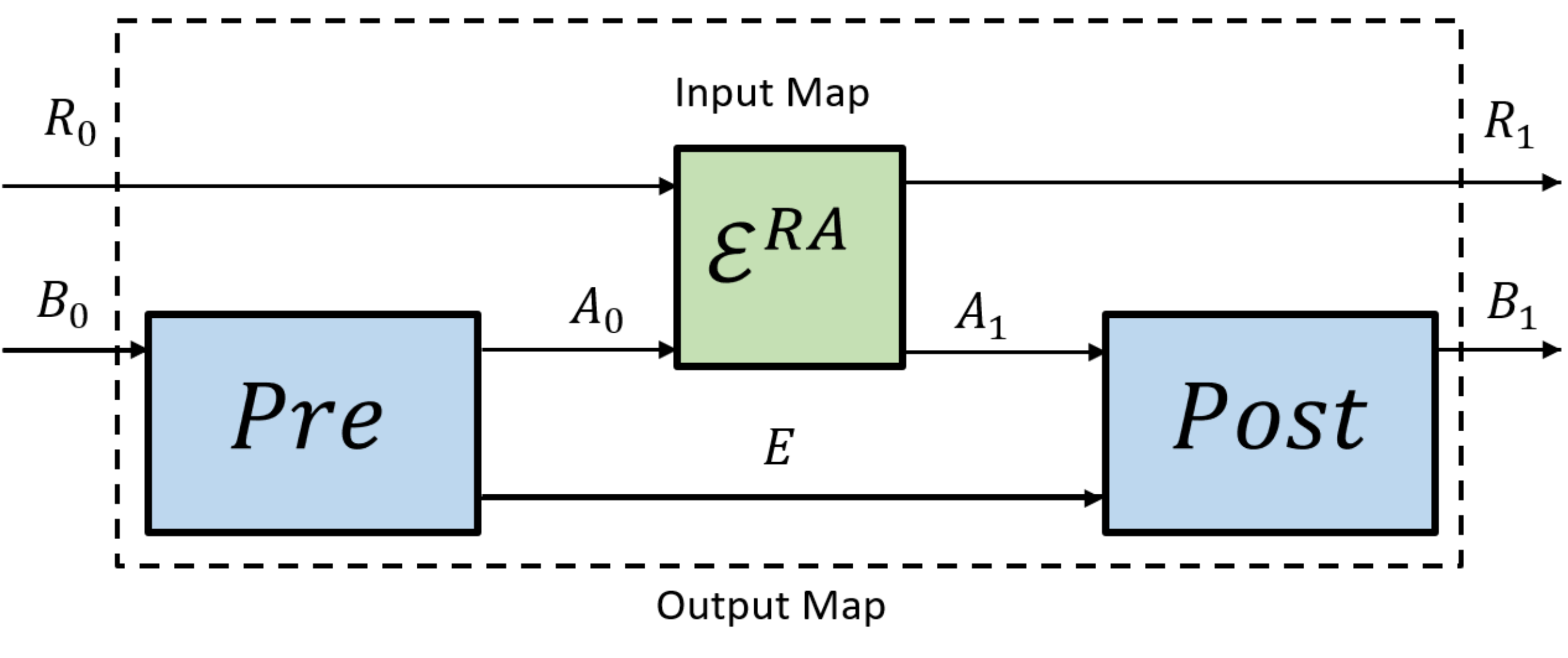}
\par\end{centering}
\caption{\label{fig:Completely}Representation of a completely CPTNI-preserving
supermap. Here, a bipartite input map $\mathcal{E}^{RA}$ is inserted
between a CPTP pre-processing map and a CPTNI post-processing map.
The output is a bipartite CPTNI map. Note the presence of the ancillary
system $E$, which acts as a ``memory'' between the pre-processing
and the post-processing. This realization of a supermap is called
a quantum 1-comb \citep{Circuit-architecture}.}
\end{figure}
 In formula, if $\Theta_{x}^{A\rightarrow B}$ is a completely CPTNI-preserving
supermap, associated with the outcome $x$ of a quantum super-instrument,
its Choi matrix $\mathbf{J}_{\Theta_{x}}^{AB}$ can be expressed in
the following form:
\begin{equation}
\mathbf{J}_{\Theta_{x}}^{AB}=\left(\mathcal{I}^{AB_{0}}\otimes\Gamma_{\mathrm{post}_{x}}^{\widetilde{A}_{1}E_{0}\rightarrow B_{1}}\right)\circ\left(\mathcal{I}^{A_{1}\widetilde{A}_{1}B_{0}}\otimes\Gamma_{\mathrm{pre}}^{\widetilde{B}_{0}\rightarrow A_{0}E_{0}}\right)\left(\phi_{+}^{B_{0}\widetilde{B}_{0}}\otimes\phi_{+}^{A_{1}\widetilde{A}_{1}}\right),\label{eq:measurementchoi}
\end{equation}
where $\mathcal{I}^{AB_{0}}$ is the identity channel on $AB_{0}$
and $\phi_{+}^{B_{0}\widetilde{B}_{0}}=\sum_{j,k}\left|jj\right\rangle \left\langle kk\right|^{B_{0}\widetilde{B}_{0}}$
is the unnormalized maximally entangled state of $B_{0}\widetilde{B}_{0}$.
Here $\Gamma_{\mathrm{pre}}^{\widetilde{B}_{0}\rightarrow A_{0}E_{0}}$
is the CPTP pre-processing map, and $\Gamma_{\mathrm{post}_{x}}^{\widetilde{A}_{1}E_{0}\rightarrow B_{1}}$
is the CPTNI post-processing map. Being that $\Theta_{x}^{A\rightarrow B}$
is part of a quantum super-instrument, we also have that $\sum_{x}\Gamma_{\mathrm{post}_{x}}^{\widetilde{A}_{1}E_{0}\rightarrow B_{1}}$
is a CPTP map. The proof of this result is reported in appendix~\ref{sec:Quantum-Super-measurements}.
Note that the pre-processing $\Gamma_{\mathrm{pre}}^{\widetilde{B}_{0}\rightarrow A_{0}E_{0}}$
is independent of $x$, therefore it can be chosen to be the same
for all the supermaps in the same super-instrument.

\section{The role of causality}

To summarize, for the first time we exactly pinned down the necessary
and sufficient conditions determining which supermaps can appear in
quantum super-instruments.  Specifically, we found that only completely
CPTNI-preserving supermaps can be implemented in a quantum super-instrument.
Additionally, we showed an explicit example of a supermap that is
CPTNI-preserving, but not completely CPTNI-preserving (Eq.~\eqref{eq:example}).
Viewing CPTNI preservation as a higher-order generalization of the
CPTNI condition for quantum maps (cf.\ Fig.~\ref{fig:CPTNI-CPTNI}),
we cannot fail to note the difference between the theory of quantum
supermaps---where CPTNI maps are regarded as states---and quantum
theory. Indeed, in quantum theory, all CPTNI maps $\mathcal{E}^{A}$
are also \emph{completely} CPTNI, the latter meaning:
\begin{equation}
\mathrm{Tr}\left[\left(\mathcal{I}^{R_{0}}\otimes\mathcal{E}^{A}\right)\left(\rho^{R_{0}A_{0}}\right)\right]\leq1,\label{eq:c-TNI}
\end{equation}
for every $\rho^{R_{0}A_{0}}\in\mathfrak{D}\left(\mathcal{H}^{R_{0}A_{0}}\right)$.
The ultimate reason for this difference is that the theory of quantum
supermaps does not satisfy the fundamental property of causality \citep{Chiribella-purification}.
\begin{ax}[Causality]
The probability of a transformation occurring in an experiment is
independent of the choice of experiments performed on its output.
\end{ax}

Loosely speaking, causality means that information cannot ``come
back from the future''. One of its consequences is that all bipartite
states are non-signalling. The existence of signalling bipartite channels
\citep{Beckman,Hoban} is a clear signature that causality does not
hold in the theory of quantum supermaps. Moreover, we can get an intuitive
grasp of the failure of causality from the realization of physical
supermaps expressed in Eq.~\eqref{eq:measurementchoi} and in Fig.~\ref{fig:Completely}.
Consider a super-measurement $\left\{ \Theta_{x}^{A\rightarrow B}\right\} $
performed on a CPTNI map $\mathcal{E}^{A}$, which means that the
measurement occurs \emph{after} $\mathcal{E}^{A}$ is prepared in
a laboratory or in a quantum circuit. The presence of pre-processing
in the realization of every $\Theta_{x}^{A\rightarrow B}$ implies
that $\Theta_{x}^{A\rightarrow B}$ acts on the \emph{input} of $\mathcal{E}^{A}$
too, meaning, in some sense, that part of $\Theta_{x}^{A\rightarrow B}$
also acts \emph{before} $\mathcal{E}^{A}$. Somehow in this situation
there is \emph{not} a well-defined notion of what comes ``before''
and ``after'', so causality cannot hold; for it would select a clear
``arrow of time'' in information processing.

More precisely, causality can be proved to be equivalent to the existence
of a unique way to discard a physical system deterministically \citep{Chiribella-purification,Coecke-no-signalling,CPT,Diagrammatic},
corresponding to the deterministic POVM $\left\{ I\right\} $ in quantum
theory, viz.\ taking the (partial) trace. However, if we want to
discard a quantum channel deterministically, we need to resort to
deterministic process POVMs \citep{ProcessPOVM}, which are highly
non-unique. More specifically, all deterministic process POVMs can
be realized as the quantum circuit fragments\[
\begin{aligned}\Qcircuit @C=1em @R=.7em @!R {&\prepareC{\rho} & \qw \poloFantasmaCn{\rA_0} &\qw & & &\qw \poloFantasmaCn{\rA_1} & \measureD{\Tr}}\end{aligned}~,
\]for any choice of quantum state $\rho^{A_{0}}\in\mathfrak{D}\left(\mathcal{H}^{A_{0}}\right)$,
where the channel is put in the slot (see appendix~\ref{subsec:The-general-framework}).
This tells us that the theory of quantum supermaps is non-causal.
In particular, for bipartite channels there are some \emph{entangled}
deterministic process POVMs:\[
\begin{aligned}\Qcircuit @C=1em @R=.7em @!R {&\multiprepareC{1}{\rho} & \qw \poloFantasmaCn{\rA_0} &\qw & & &\qw \poloFantasmaCn{\rA_1} & \measureD{\Tr} \\ &\pureghost{\rho} & \qw \poloFantasmaCn{\rB_0} &\qw & & &\qw \poloFantasmaCn{\rB_1} & \measureD{\Tr}}\end{aligned}~,
\]for $\rho^{A_{0}B_{0}}\in\mathfrak{D}\left(\mathcal{H}^{A_{0}B_{0}}\right)$
entangled. We can collectively call the deterministic ways to discard
a physical object (whether a quantum state or a quantum channel) as
`deterministic effects', which will be denoted as $u$ in circuit
diagrams. It can be shown that \emph{no} entangled deterministic effects
exist in causal theories, where they are in a tensor product form
\citep{Chiribella-purification}. For example, in quantum theory the
deterministic POVM on the composite system $\mathcal{H}^{A_{0}}\otimes\mathcal{H}^{B_{0}}$
is $I^{A_{0}B_{0}}=I^{A_{0}}\otimes I^{B_{0}}$.

Using operational probabilistic theories \citep{Chiribella-purification,Chiribella-informational,hardy2011,Hardy-informational-2},
based on the notion of circuits and on the composition of physical
transformations occurring in experiments (see appendix~\ref{subsec:The-general-framework}),
we can achieve a unified approach to establishing which operations
are physical. These will be the operations on the objects of interest
(i.e.\ states or channels) that can arise in a measurement or in
a generic experiment done on them. In particular, we obtain the following
necessary and sufficient condition valid in any unrestricted theory
(e.g.\ when there are no superselection rules), whether causal or
not.
\begin{thm}
$\mathcal{A}$ is a physical operation if and only if it is `completely
positive' and\begin{equation}\label{eq:OPT physical}
\begin{aligned}\Qcircuit @C=1em @R=.7em @!R { & \multiprepareC{1}{\rho} & \qw \poloFantasmaCn{\rA} & \gate{\cA} & \qw \poloFantasmaCn{\rB} & \multimeasureD{1}{u} \\ & \pureghost{\rho} & \qw \poloFantasmaCn{\rS} & \qw &\qw &\ghost{u} }\end{aligned} ~\leq 1,
\end{equation}for every system $\mathrm{S}$, every `deterministic state' $\rho$,
and every deterministic effect $u$.
\end{thm}

Here, by `deterministic state' we mean the deterministic object
we are interested in; it can be a quantum state or a quantum channel.
$\mathcal{A}$ represents an operation that is performed on the objects
of interest: a quantum map in quantum theory, or a quantum supermap
if our objects of interest are quantum channels. `Completely positive'
in this setting means that the operation preserves the convex cone
generated by `states' in a complete sense. For quantum theory, it
means that $\mathcal{A}$ is CP; for quantum supermaps that $\mathcal{A}$
is a CPP supermap. Now, let us explain the meaning of Eq.~\eqref{eq:OPT physical}.
In quantum theory, $u$ is nothing but the deterministic POVM, i.e.\ tracing
out both systems, which has a tensor product form due to causality.
We obtain the simple TNI condition $\mathrm{Tr}\left[\mathcal{A}\left(\rho\right)\right]\leq1$
for all quantum states $\rho$, and not Eq.~\eqref{eq:c-TNI}. If
our objects of interest are, instead, quantum channels, this time
$u$ is a deterministic bipartite process POVM. Now, the existence
of entangled bipartite process POVMs, a consequence of the \emph{failure
of causality}, does \emph{not} allow us to reduce the condition of
Eq.~\eqref{eq:OPT physical} to a simple condition involving only
one system. Therefore, the conditions of `complete positivity' and
Eq.~\eqref{eq:OPT physical} imply that $\mathcal{A}$ must be CPP
and \emph{completely} CPTNI preserving (Eq.~\eqref{eq:c-CPTNI}),
and we cannot get rid of the `complete sense' required in this condition.

A rigorous proof and statement of this result, along with more details
on the implications of causality for the theory of quantum supermaps,
is presented in appendix~\ref{subsec:physical maps}.

\section{Conclusion}

The results we obtained in this article improve our understanding
of the operational viewpoint in quantum theory, and more generally
in physics. In particular, we showed that the correct conditions to
impose on a linear transformation to guarantee its physicality, be
it a quantum map or a quantum supermap, must always be formulated
in a \emph{complete sense}. This means that they must always involve
the tensor product with the identity transformation. Thus, for quantum
supermaps we have the CPP condition and the complete CPTNI preservation
condition. For quantum maps we have the CP condition and the complete
TNI condition of Eq.~\eqref{eq:c-TNI}. Since quantum theory satisfies
causality, Eq.~\eqref{eq:c-TNI} becomes \emph{equivalent} to the
TNI condition we impose ordinarily on quantum maps. However, the fundamental
requirement is still the one expressed by Eq.~\eqref{eq:c-TNI}.
In other words, the existence of signalling bipartite states in the
theory of quantum supermaps is the reason for the gap between CPTNI
preservation and complete CPTNI preservation; in the very same way
as the existence of entangled states in quantum theory creates a gap
between positive and completely positive maps, which are instead the
same notion in classical physics.

The fact that conditions expressed in a complete sense are the right
thing to demand is apparent if one adopts the circuit framework in
which operational probabilistic theories are formulated. Our results
confirm and strengthen the validity of this approach to the study
of the fundamental operational properties of physical theories.

\medskip{}

\paragraph{Contributions}

JB, MG, and CMS did most of the technical work, and contributed equally
to the article. CMS wrote most of the article, and had the idea of
linking the main result with the principle of causality. GC and GG
had the main idea of the article, and provided guidance and direction
to the research.

\begin{acknowledgments}
The authors acknowledge support from the Natural Sciences and Engineering
Research Council of Canada (NSERC). CMS acknowledges support from
the Pacific Institute for the Mathematical Sciences (PIMS) and from
a Faculty of Science Grand Challenge award at the University of Calgary.
GC is supported by the National Natural Science Foundation of China
through grant 11675136, the Foundational Questions Institute through
grant FQXi-RFP3-1325, the Hong Kong Research Grant Council through
grant 17300918, and the John Templeton Foundation through grant 60609,
Quantum Causal Structures. The opinions expressed in this publication
are those of the authors and do not necessarily reflect the views
of the John Templeton Foundation. 
\end{acknowledgments}

\bibliographystyle{apsrev4-1}
\bibliography{bib}

\appendix

\section{General facts about quantum maps and supermaps\label{sec:General-facts-about}}

Quantum maps describe the evolution of quantum systems, in both the
deterministic and the probabilistic case (e.g.\ when a measurement
is performed). To be consistent with the interpretation of mixed states
as probabilistic ensembles, a quantum map $\mathcal{E}^{A}$ must
be linear, $\mathcal{E}^{A}\in\mathfrak{L}^{A}$ . This is not enough,
because a quantum map must send quantum states to quantum states even
when applied only to half of a bipartite state. For this reason we
first demand that it be \emph{completely positive} (CP): for every
$\rho^{R_{0}A_{0}}\in\mathfrak{D}\left(\mathcal{H}^{R_{0}A_{0}}\right)$
it must be
\[
\left(\mathcal{I}^{R_{0}}\otimes\mathcal{E}^{A}\right)\left(\rho^{R_{0}A_{0}}\right)\geq0,
\]
where $\mathcal{I}^{R_{0}}$ is the identity map on system $R_{0}$.
This means that $\mathcal{E}^{A}$ sends positive semi-definite operators
to positive semi-definite operators even when tensored with the identity.
We also require that a map $\mathcal{E}^{A}$ be \emph{trace non-increasing}
(TNI): 
\[
\mathrm{Tr}\left[\mathcal{E}^{A}\left(\rho^{A_{0}}\right)\right]\leq1,
\]
for every $\rho^{A_{0}}\in\mathfrak{D}\left(\mathcal{H}^{A_{0}}\right)$.
In particular, if the trace is preserved, that is $\mathrm{Tr}\left[\mathcal{E}^{A}\left(\rho^{A_{0}}\right)\right]=1$,
for every $\rho^{A_{0}}\in\mathfrak{D}\left(\mathcal{H}^{A_{0}}\right)$,
we say that the map is \emph{trace-preserving} (TP). The allowed quantum
maps are those that are both CP and TNI (CPTNI). CPTP maps are also
called \emph{quantum channels}, and represent the most general deterministic
evolutions a quantum system can undergo. CPTNI maps that are not CPTP
represent non-deterministic transformations. This is what happens
in a quantum measurement, which can be seen as a collection of CPTNI
maps $\left\{ \mathcal{E}_{x}^{A}\right\} $, indexed by the outcomes
$x$ of that measurement, such that $\sum_{x}\mathcal{E}_{x}^{A}$
is a CPTP map. If we know the outcome $x$ of the measurement, then
we know that the system evolved under the CPTNI map $\mathcal{E}_{x}^{A}$.
We can therefore construct a \emph{quantum instrument}
\begin{equation}
\mathcal{E}^{A_{0}\rightarrow X_{1}A_{1}}=\sum_{x}\left|x\right\rangle \left\langle x\right|^{X_{1}}\otimes\mathcal{E}_{x}^{A},\label{eq:quantum instrument}
\end{equation}
where $\left\{ \left|x\right\rangle ^{X_{1}}\right\} $ is an orthonormal
basis of system $X_{1}$. $\mathcal{E}^{A_{0}\rightarrow X_{1}A_{1}}$
is a quantum channel with classical-quantum output. Here $X_{1}$
is the classical system, recording the measurement outcome. As such,
it represents the meter read by the experimenter performing the quantum
measurement $\left\{ \mathcal{E}_{x}^{A}\right\} $. 

These notions can be easily generalized to quantum supermaps \citep{Chiribella2008,Hierarchy-combs,Perinotti1},
namely to transformations sending quantum maps to quantum maps. Again,
these are linear maps, and an easy translation of the requirements
of CP and TNI leads to the requirement of CPP (Eq.~\eqref{eq:CPP}
in the main article) \citep{Chiribella2008,Gour2018} and TNI preservation.
Specifically, a map is CPTNI-preserving if it is CPP, and sends CPTNI
maps to CPTNI maps:
\begin{equation}
\mathrm{Tr}\left[\Theta^{A\rightarrow B}\left[\mathcal{E}^{A}\right]\left(\rho^{B_{0}}\right)\right]\leq1,\label{eq:CPTNI true}
\end{equation}
for any CPTNI map $\mathcal{E}^{A}$ and any $\rho^{B_{0}}\in\mathfrak{D}\left(\mathcal{H}^{B_{0}}\right)$.
In fact, if $\Theta^{A\rightarrow B}$ is CPP, it is enough to require
that inequality~\eqref{eq:CPTNI true} be satisfied by quantum channels
$\mathcal{E}^{A}$, namely by CPTP maps.

To see it, let $\mathcal{E}^{A}$ be a CPTNI map. We can always find
another CPTNI map $\mathcal{E}'^{A}$ such that $\mathcal{E}^{A}+\mathcal{E}'^{A}$
is CPTP. Now assume that $\Theta^{A\rightarrow B}$ is CPP, and sends
CPTP maps to CPTNI maps. Then, for every $\rho^{B_{0}}\in\mathcal{D}\left(\mathcal{H}^{B_{0}}\right)$,
\[
1\geq\mathrm{Tr}\left[\Theta^{A\rightarrow B}\left[\mathcal{E}^{A}+\mathcal{E}'^{A}\right]\left(\rho^{B_{0}}\right)\right]=\mathrm{Tr}\left[\Theta^{A\rightarrow B}\left[\mathcal{E}^{A}\right]\left(\rho^{B_{0}}\right)\right]+\mathrm{Tr}\left[\Theta^{A\rightarrow B}\left[\mathcal{E}'^{A}\right]\left(\rho^{B_{0}}\right)\right].
\]
Since $\Theta^{A\rightarrow B}$ is CPP, then $\mathrm{Tr}\left[\Theta^{A\rightarrow B}\left[\mathcal{E}^{A}\right]\left(\rho^{B_{0}}\right)\right]\geq0$
and $\mathrm{Tr}\left[\Theta^{A\rightarrow B}\left[\mathcal{E}'^{A}\right]\left(\rho^{B_{0}}\right)\right]\geq0$,
therefore we conclude that it must be
\[
\mathrm{Tr}\left[\Theta^{A\rightarrow B}\left[\mathcal{E}^{A}\right]\left(\rho^{B_{0}}\right)\right]\leq1,
\]
which means that $\Theta^{A\rightarrow B}$ satisfies Eq.~\eqref{eq:CPTNI true}.

A CPTNI-preserving supermap $\Theta^{A\rightarrow B}$ is called \emph{superchannel}
if it sends CPTP maps to CPTP maps \citep{Gour2018}. The original
definition in \citep{Switch} required that it should send quantum
channels to quantum channels in a complete sense, i.e.\ even when
tensored with the identity supermap. In other words, 
\[
\mathrm{Tr}\left[\left(\mathbf{1}^{R}\otimes\Theta^{A\rightarrow B}\right)\left[\mathcal{N}^{RA}\right]\left(\rho^{R_{0}B_{0}}\right)\right]=1,
\]
for any CPTP map $\mathcal{N}^{RA}$ and any $\rho^{R_{0}B_{0}}\in\mathfrak{D}\left(\mathcal{H}^{R_{0}B_{0}}\right)$,
where $\mathbf{1}^{R}$ is the identity supermap on $R$. Actually,
in \citep[theorem 1]{Gour2018}, using the Choi picture, it was proved
that we need not consider this requirement in a complete sense: a
CPTNI-preserving supermap $\Theta^{A\rightarrow B}$ is a superchannel
if and only if
\begin{equation}
\mathrm{Tr}\left[\Theta^{A\rightarrow B}\left[\mathcal{N}^{A}\right]\left(\rho^{B_{0}}\right)\right]=1,\label{eq:superchannel trace}
\end{equation}
for any CPTP map $\mathcal{N}^{A}$ and any $\rho^{B_{0}}\in\mathfrak{D}\left(\mathcal{H}^{B_{0}}\right)$.
In Appendix~\ref{subsec:physical maps} we will prove this result
in an alternative way, without using the Choi isomorphism.

Superchannels are intimately related to channels: it was proved that
all superchannels can be represented in terms of a pre- and a post-processing
CPTP map \citep{Chiribella2008,Gour2018}, as depicted in Fig.~\ref{fig:superchannel}.
\begin{figure}
\begin{centering}
\includegraphics[scale=0.5]{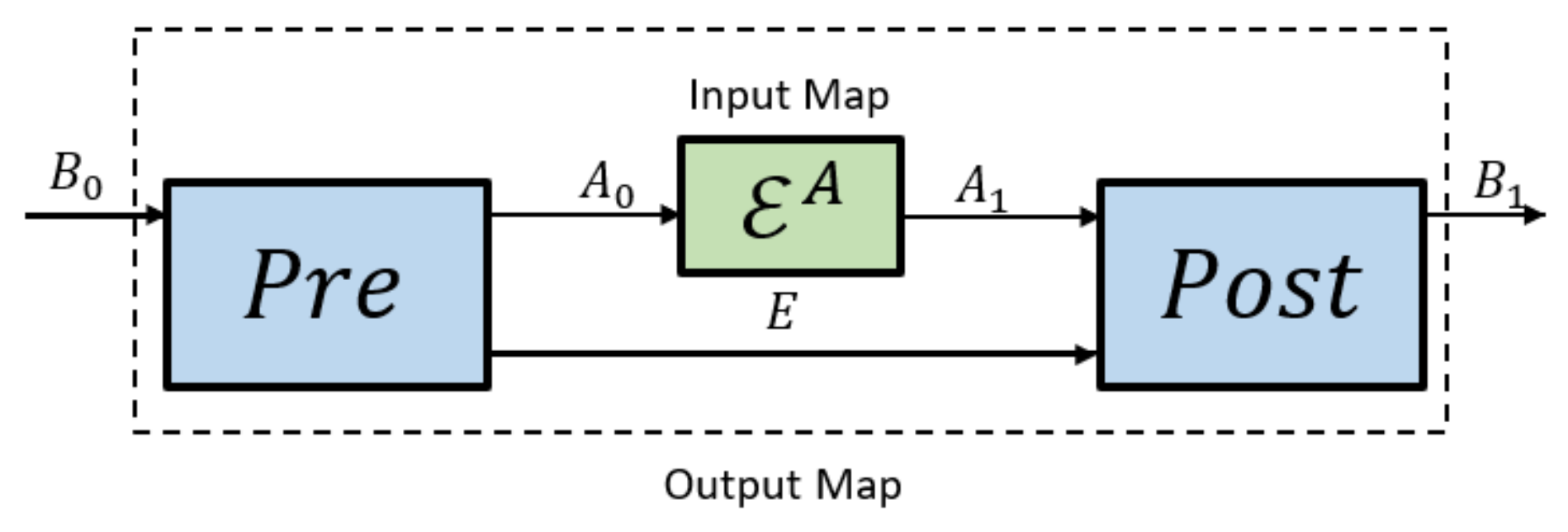}
\par\end{centering}
\caption{\label{fig:superchannel}Realization of a superchannel. Here an input
quantum channel is inserted between a CPTP pre-processing map and
a CPTP post-processing map. The output is another quantum channel.
The ancillary system $E$ plays the role of a quantum memory between
the pre- and the post-processing.}
\end{figure}
 Such a representation is called a quantum 1-comb \citep{Chiribella2008}.

Superchannels play an important role because they represent all physical
ways a quantum channel can evolve in an open system, and provide a
framework for measurements on quantum operations, called \emph{super-measurements}.
These are described by a set $\left\{ \Theta_{x}^{A\rightarrow B}\right\} $
of CPTNI-preserving supermaps such that $\sum_{x}\Theta_{x}^{A\rightarrow B}$
is a superchannel. Then we can construct a \emph{quantum super-instrument}
as the generalization of the quantum notion (see Eq.~\eqref{eq:super-instrument}
in the main article):
\[
\Upsilon^{A\rightarrow X_{1}B}\left[\mathcal{E}^{A}\right]=\sum_{x}\left|x\right\rangle \left\langle x\right|^{X_{1}}\otimes\Theta_{x}^{A\rightarrow B}\left[\mathcal{E}^{A}\right],
\]
where system $X_{1}$ again represents the classical meter, as in
Eq.~\eqref{eq:quantum instrument}, and $\mathcal{E}^{A}$ is any
CP map. The main result of this article is that, unlike CPTNI quantum
maps, \emph{not all} CPTNI-preserving supermaps can be part of a quantum
super-instrument. In Appendix~\ref{sec:OPT-interpretation-of} we
link this fact to the failure of causality \citep{Chiribella-purification}
in the theory of quantum supermaps.

\section{The Choi picture for quantum maps and supermaps}

In this appendix we collect some results about the Choi representation
of quantum maps and supermaps. Specifically, Appendix~\ref{subsec:Useful-results-about}
provides the general background information and some basic results
about the Choi isomorphism. Appendix~\ref{subsec:Some-technical-derivations}
instead focuses on the derivations related to complete CPTNI preservation
in the Choi form, and in particular on obtaining Eq.~\eqref{eq:completechoicond}
in the main article.

\subsection{Choi matrices of quantum maps and supermaps\label{subsec:Useful-results-about}}

The first ingredient to define the Choi isomorphism is to consider
the super-normalized maximally entangled state $\left|\phi_{+}\right\rangle ^{A_{0}\widetilde{A}_{0}}=\sum_{j=1}^{\left|A_{0}\right|}\left|j\right\rangle ^{A_{0}}\left|j\right\rangle ^{\widetilde{A}_{0}}$,
where $\left\{ \left|j\right\rangle ^{A_{0}}\right\} $ is a \emph{fixed}
orthonormal basis of $\mathcal{H}^{A_{0}}$ (and therefore of $\mathcal{H}^{\widetilde{A}_{0}}$
too, since $\widetilde{A}_{0}$ is just another copy of $A_{0}$).
The Choi matrix of a linear map $\mathcal{E}^{A}\in\mathfrak{L}^{A}$
is defined as
\[
J_{\mathcal{E}}^{A}:=\left(\mathcal{I}^{A_{0}}\otimes\mathcal{E}^{\widetilde{A}_{0}\rightarrow A_{1}}\right)\left(\phi_{+}^{A_{0}\widetilde{A}_{0}}\right),
\]
where $\phi_{+}^{A_{0}\widetilde{A}_{0}}:=\left|\phi_{+}\right\rangle \left\langle \phi_{+}\right|^{A_{0}\widetilde{A}_{0}}$,
and $\mathcal{I}^{A_{0}}$ is the identity channel. Again, since $\widetilde{A}_{0}$
is just another copy of $A_{0}$, the linear map $\mathcal{E}^{\widetilde{A}_{0}\rightarrow A_{1}}$
is well defined.

In particular, $\mathcal{E}^{A}$ is CP if and only if $J_{\mathcal{E}}^{A}\geq0$.
$\mathcal{E}^{A}$ is CPTP if and only if in addition one has $J_{\mathcal{E}}^{A_{0}}=I^{A_{0}}$.
Instead, $\mathcal{E}^{A}$ is CPTNI if and only if, besides $J_{\mathcal{E}}^{A}\geq0$,
one has $J_{\mathcal{E}}^{A_{0}}\leq I^{A_{0}}$. The Choi matrix
$J_{\mathcal{E}}^{A}$ encodes all the information about $\mathcal{E}^{A}$
because one can reconstruct the action of $\mathcal{E}^{A}$ on quantum
states from its Choi matrix:
\begin{equation}
\mathcal{E}^{A}\left(\rho^{A_{0}}\right)=\mathrm{Tr}_{A_{0}}\left[J_{\mathcal{E}}^{A}\left(\left(\rho^{A_{0}}\right)^{\mathrm{T}}\otimes I^{A_{1}}\right)\right],\label{eq:Choi on states}
\end{equation}
for every $\rho^{A_{0}}\in\mathfrak{D}\left(\mathcal{H}^{A_{0}}\right)$.

To define the Choi matrix of a supermap $\Theta^{A\rightarrow B}$,
we follow the approach presented in \citep{Gour2018}. Let us consider
the following basis of the space $\mathfrak{L}^{A}$:
\[
\mathcal{E}_{jklm}^{A}\left(\rho^{A_{0}}\right)=\left\langle j\middle|\rho\middle|k\right\rangle ^{A_{0}}\left|l\right\rangle \left\langle m\right|^{A_{1}},
\]
for $j,k\in\left\{ 1,\ldots,\left|A_{0}\right|\right\} $ and $l,m\in\left\{ 1,\ldots,\left|A_{1}\right|\right\} $.
The Choi matrix of the supermap $\Theta^{A\rightarrow B}$ can be
defined as
\[
\mathbf{J}_{\Theta}^{AB}:=\sum_{j,k,l,m}J_{\mathcal{E}_{jklm}}^{A}\otimes J_{\Theta\left[\mathcal{E}_{jklm}\right]}^{B}.
\]
Again, $\mathbf{J}_{\Theta}^{AB}$ encodes all the information about
$\Theta^{A\rightarrow B}$. For instance, $\Theta^{A\rightarrow B}$
is CPP if and only if $\mathbf{J}_{\Theta}^{AB}\geq0$. Moreover,
we can express the action of a supermap on a quantum map $\mathcal{E}^{A}$
using their Choi matrices: if $\mathcal{F}^{B}=\Theta^{A\rightarrow B}\left[\mathcal{E}^{A}\right]$,
we have \citep{Gour2018}:
\begin{equation}
J_{\mathcal{F}}^{B}=\mathrm{Tr}_{A}\left[\mathbf{J}_{\Theta}^{AB}\left(\left(J_{\mathcal{E}}^{A}\right)^{\mathrm{T}}\otimes I^{B}\right)\right].\label{eq:Choi on maps}
\end{equation}

A full characterization of superchannels from their Choi matrices
was given in \citep{Gour2018}: $\Theta^{A\rightarrow B}$ is a superchannel
if and only if $\mathbf{J}_{\Theta}^{AB}\ge0$, and one has $\mathbf{J}_{\Theta}^{AB_{0}}=\mathbf{J}_{\Theta}^{A_{0}B_{0}}\otimes u^{A_{1}}$
and $\mathbf{J}_{\Theta}^{A_{1}B_{0}}=I^{A_{1}B_{0}}.$ Here $u^{A_{1}}=\frac{1}{\left|A_{1}\right|}I^{A_{1}}$
is the maximally mixed state. Combining Eqs.~\eqref{eq:Choi on states}
and \eqref{eq:Choi on maps} for a CPTP map $\mathcal{N}^{A}$, we
have
\[
\Theta^{A\rightarrow B}\left[\mathcal{N}^{A}\right]\left(\rho^{B_{0}}\right)=\mathrm{Tr}_{AB_{0}}\left[\mathbf{J}_{\Theta}^{AB}\left(\left(J_{\mathcal{N}}^{A}\otimes\rho^{B_{0}}\right)^{\mathrm{T}}\otimes I^{B_{1}}\right)\right];
\]
therefore,
\begin{equation}
\mathrm{Tr}\left[\Theta^{A\rightarrow B}\left[\mathcal{N}^{A}\right]\left(\rho^{B_{0}}\right)\right]=\mathrm{Tr}_{AB_{0}}\left\{ \mathrm{Tr}_{B_{1}}\left[\mathbf{J}_{\Theta}^{AB}\left(\left(J_{\mathcal{N}}^{A}\otimes\rho^{B_{0}}\right)^{\mathrm{T}}\otimes I^{B_{1}}\right)\right]\right\} =\mathrm{Tr}\left[\mathbf{J}_{\Theta}^{AB_{0}}\left(J_{\mathcal{N}}^{A}\otimes\rho^{B_{0}}\right)^{\mathrm{T}}\right].\label{eq:trace}
\end{equation}
Hence, by Eq.~\eqref{eq:superchannel trace} $\Theta^{A\rightarrow B}$
is a superchannel if and only if
\[
\mathrm{Tr}\left[\mathbf{J}_{\Theta}^{AB_{0}}\left(J_{\mathcal{N}}^{A}\otimes\rho^{B_{0}}\right)^{\mathrm{T}}\right]=1.
\]

Similarly, we can characterize CPTNI-preserving supermaps in the Choi
picture. By Eq.~\eqref{eq:CPTNI} in the main article, a supermap
is CPTNI-preserving if $\mathrm{Tr}\left[\Theta^{A\rightarrow B}\left[\mathcal{N}^{A}\right]\left(\rho^{B_{0}}\right)\right]\leq1$,
for every CPTP map $\mathcal{N}^{A}$ and every density matrix $\rho^{B_{0}}$.
By Eq.~\eqref{eq:trace}, we can rewrite this condition in the Choi
picture as
\[
\mathrm{Tr}\left[\mathbf{J}_{\Theta}^{AB_{0}}\left(J_{\mathcal{N}}^{A}\otimes\rho^{B_{0}}\right)^{\mathrm{T}}\right]\leq1.
\]
This proves Eq.~\eqref{eq:choicond} in the main article.

\subsection{Some technical derivations about completely CPTNI-preserving supermaps\label{subsec:Some-technical-derivations}}

Now we will focus on expressing the complete CPTNI preservation condition
in the Choi picture. Looking at Eq.~\eqref{eq:c-CPTNI} in the main
article tells us that we need to find an expression for $\mathrm{Tr}\left[\left(\mathbf{1}^{R}\otimes\Theta^{A\rightarrow B}\right)\left[\mathcal{N}^{RA}\right]\left(\rho^{R_{0}B_{0}}\right)\right]$
in the Choi picture. Note that the identity supermap does not change
the systems it acts on. Therefore, to express Eq.~\eqref{eq:c-CPTNI}
in the main article in the Choi form, we only consider how $\mathcal{N}^{RA}$
is acted on by the supermap $\Theta^{A\rightarrow B}$, representing
the action of the identity superchannel with the identity matrix $I^{R}$.
Therefore, combining Eqs.~\eqref{eq:Choi on states} and \eqref{eq:Choi on maps}
this time yields:
\begin{eqnarray}
\mathrm{Tr}\left[\left(\mathbf{1}^{R}\otimes\Theta^{A\rightarrow B}\right)\left[\mathcal{N}^{RA}\right]\left(\rho^{R_{0}B_{0}}\right)\right] & = & \mathrm{Tr}\left[\left(I^{R}\otimes\mathbf{J}_{\Theta}^{AB}\right)\left(\left(J_{\mathcal{N}}^{RA}\right)^{\mathrm{T}_{A}}\otimes I^{B}\right)\left(\left(\rho^{R_{0}B_{0}}\right)^{\mathrm{T}}\otimes I^{R_{1}AB_{1}}\right)\right]\nonumber \\
 & = & \mathrm{Tr}_{R_{0}AB_{0}}\left\{ \mathrm{Tr}_{R_{1}B_{1}}\left[\left(I^{R}\otimes\mathbf{J}_{\Theta}^{AB}\right)\left(\left(J_{\mathcal{N}}^{RA}\right)^{\mathrm{T}_{A}}\otimes I^{B}\right)\left(\left(\rho^{R_{0}B_{0}}\right)^{\mathrm{T}}\otimes I^{R_{1}AB_{1}}\right)\right]\right\} \nonumber \\
 & = & \mathrm{Tr}_{R_{0}AB_{0}}\left[\left(I^{R_{0}}\otimes\mathbf{J}_{\Theta}^{AB_{0}}\right)\left(\left(J_{\mathcal{N}}^{R_{0}A}\right)^{\mathrm{T}_{A}}\otimes I^{B_{0}}\right)\left(\left(\rho^{R_{0}B_{0}}\right)^{\mathrm{T}}\otimes I^{A}\right)\right].\label{eq:c-CPTNI trace}
\end{eqnarray}
Now let us define 
\begin{equation}
M^{AB_{0}}:=\mathrm{Tr}_{R_{0}}\left[\left(\rho^{R_{0}B_{0}}\otimes I^{A}\right)\left(\left(J_{\mathcal{N}}^{R_{0}A}\right)^{\mathrm{T}_{R_{0}}}\otimes I^{B_{0}}\right)\right],\label{eq: M def}
\end{equation}
and let us calculate:
\begin{align*}
\mathrm{Tr}\left[\mathbf{J}_{\Theta}^{AB_{0}}\left(M^{AB_{0}}\right)^{\mathrm{T}}\right] & =\mathrm{Tr}_{R_{0}AB_{0}}\left[\left(I^{R_{0}}\otimes\mathbf{J}_{\Theta}^{AB_{0}}\right)\left(\left(J_{\mathcal{N}}^{R_{0}A}\right)^{\mathrm{T}_{R_{0}}}\otimes I^{B_{0}}\right)^{\mathrm{T}}\left(\rho^{R_{0}B_{0}}\otimes I^{A}\right)^{\mathrm{T}}\right]\\
 & =\mathrm{Tr}\left[\left(I^{R_{0}}\otimes\mathbf{J}_{\Theta}^{AB_{0}}\right)\left(\left(J_{\mathcal{N}}^{R_{0}A}\right)^{\mathrm{T}_{A}}\otimes I^{B_{0}}\right)\left(\left(\rho^{R_{0}B_{0}}\right)^{\mathrm{T}}\otimes I^{A}\right)\right].
\end{align*}
As we can see, this coincides with Eq.~\eqref{eq:c-CPTNI trace}.
Therefore $\mathrm{Tr}\left[\left(\mathbf{1}^{R}\otimes\Theta^{A\rightarrow B}\right)\left[\mathcal{N}^{RA}\right]\left(\rho^{R_{0}B_{0}}\right)\right]=\mathrm{Tr}\left[\mathbf{J}_{\Theta}^{AB_{0}}\left(M^{AB_{0}}\right)^{\mathrm{T}}\right]$,
where $M^{AB_{0}}$ is defined in Eq.~\eqref{eq: M def}. Now the
complete CPTNI preservation condition of Eq.~\eqref{eq:c-CPTNI}
in the main article becomes:
\begin{equation}
\mathrm{Tr}\left[\mathbf{J}_{\Theta}^{AB_{0}}\left(M^{AB_{0}}\right)^{\mathrm{T}}\right]\leq1,\label{eq:completechoicond-partial}
\end{equation}
for every $M^{AB_{0}}$ of the form~\eqref{eq: M def}. Note that
$\mathrm{Tr}\left[\mathbf{J}_{\Theta}^{AB_{0}}\left(M^{AB_{0}}\right)^{\mathrm{T}}\right]\geq0$
for every CPP supermap $\Theta^{A\rightarrow B}$, whence $M^{AB_{0}}$
is positive semi-definite. Furthermore,
\[
M^{A_{0}B_{0}}=\mathrm{Tr}_{R_{0}A_{1}}\left[\left(\rho^{R_{0}B_{0}}\otimes I^{A}\right)\left(\left(J_{\mathcal{N}}^{R_{0}A}\right)^{\mathrm{T}_{R_{0}}}\otimes I^{B_{0}}\right)\right]=\mathrm{Tr}_{R_{0}}\left[\rho^{R_{0}B_{0}}\right]\otimes I^{A_{0}}=I^{A_{0}}\otimes\rho^{B_{0}},
\]
where we have used the fact that $\mathrm{Tr}_{R_{0}A_{1}}\left[\left(J_{\mathcal{N}}^{R_{0}A}\right)^{\mathrm{T}_{R_{0}}}\right]=\mathrm{Tr}_{R_{0}A_{1}}\left[J_{\mathcal{N}}^{R_{0}A}\right]$,
and that $\mathrm{Tr}_{A_{1}}\left[J_{\mathcal{N}}^{R_{0}A}\right]=\mathrm{Tr}_{R_{1}A_{1}}\left[J_{\mathcal{N}}^{RA}\right]=I^{R_{0}}\otimes I^{A_{0}}$
because $\mathcal{N}$ is CPTP (cf.\ Appendix~\ref{subsec:Useful-results-about}).
So $M^{AB_{0}}$ has marginal $M^{A_{0}B_{0}}=I^{A_{0}}\otimes\rho^{B_{0}}$.

Now we prove a key result, namely that \emph{every} positive semi-definite
matrix $M^{AB_{0}}$ with marginal $M^{A_{0}B_{0}}=I^{A_{0}}\otimes\rho^{B_{0}}$,
where $\rho^{B_{0}}$ is any density matrix, can be written as in
Eq.~\eqref{eq: M def}. In this way, instead of stating complete
CPTNI preservation as in Eq.~\eqref{eq:completechoicond-partial}
for $M^{AB_{0}}$ of the form~\eqref{eq: M def}, we will state it
in a remarkably simpler way: $\Theta^{A\rightarrow B}$ is completely
CPTNI-preserving if and only if Eq.~\eqref{eq:completechoicond-partial}
is satisfied for \emph{any} positive semi-definite $M^{AB_{0}}$ with
marginal $M^{A_{0}B_{0}}=I^{A_{0}}\otimes\rho^{B_{0}}$. This technical
result will be crucial for the main finding of this article, namely
the characterization of physical supermaps (see Appendix~\ref{sec:Completion-of-Supermaps}).
\begin{lem}
Let $M^{AB_{0}}\geq0$ such that $M^{A_{0}B_{0}}=I^{A_{0}}\otimes\rho^{B_{0}}$,
for some $\rho^{B_{0}}\in\mathfrak{D}\left(\mathcal{H}^{B_{0}}\right)$.
Then
\[
M^{AB_{0}}=\mathrm{Tr}_{R_{0}}\left[\left(\rho^{R_{0}B_{0}}\otimes I^{A}\right)\left(\left(J_{\mathcal{N}}^{R_{0}A}\right)^{\mathrm{T}_{R_{0}}}\otimes I^{B_{0}}\right)\right],
\]
where $\mathcal{N}^{RA}$ is some CPTP map and $\rho^{R_{0}B_{0}}\in\mathfrak{D}\left(\mathcal{H}^{R_{0}B_{0}}\right)$.
\end{lem}

\begin{proof}
Let $\phi_{+}^{A_{0}\widetilde{A}_{0}}\otimes\varphi^{E_{0}B_{0}}$
be a purification of $M^{A_{0}B_{0}}=I^{A_{0}}\otimes\rho^{B_{0}}$,
where $\varphi^{E_{0}B_{0}}\in\mathfrak{D}\left(\mathcal{H}^{E_{0}B_{0}}\right)$
is a purification of $\rho^{B_{0}}$. Now let $\tau^{AB_{0}F_{0}}$
be a purification of $M^{AB_{0}}$, so $\tau^{AB_{0}F_{0}}$ is also
a purification of $M^{A_{0}B_{0}}$. Thus, these two purifications
can be related by an isometry channel $\mathcal{V}^{\widetilde{A}_{0}E_{0}\rightarrow A_{1}F_{0}}$
such that \citep{book}
\[
\tau^{AB_{0}F_{0}}=\left(\mathcal{I}^{A_{0}B_{0}}\otimes\mathcal{V}^{\widetilde{A}_{0}E_{0}\rightarrow A_{1}F_{0}}\right)\left(\phi_{+}^{A_{0}\widetilde{A}_{0}}\otimes\varphi^{E_{0}B_{0}}\right).
\]
Performing the partial trace on system $F_{0}$ yields
\begin{equation}
M^{AB_{0}}=\left(\mathcal{I}^{A_{0}B_{0}}\otimes\Gamma^{\widetilde{A}_{0}E_{0}\rightarrow A_{1}}\right)\left(\phi_{+}^{A_{0}\widetilde{A}_{0}}\otimes\varphi^{E_{0}B_{0}}\right),\label{Nforminter}
\end{equation}
where $\Gamma^{\widetilde{A}_{0}E_{0}\rightarrow A_{1}}:=\mathrm{Tr}_{F_{0}}\circ\mathcal{V}^{\widetilde{A}_{0}E_{0}\rightarrow A_{1}F_{0}}$
is a CPTP map. The action of $\Gamma^{\widetilde{A}_{0}E_{0}\rightarrow A_{1}}$
on a generic state $\chi^{\widetilde{A}_{0}E_{0}}\in\mathfrak{D}\left(\mathcal{H}^{\widetilde{A}_{0}E_{0}}\right)$
can be written in terms of its Choi matrix as
\begin{equation}
\Gamma^{\widetilde{A}_{0}E_{0}\rightarrow A_{1}}\left(\chi^{\widetilde{A}_{0}E_{0}}\right)=\mathrm{Tr}_{\widetilde{A}_{0}E_{0}}\left[J_{\Gamma}^{\widetilde{A}_{0}E_{0}A_{1}}\left(\left(\chi^{\widetilde{A}_{0}E_{0}}\right)^{\mathrm{T}}\otimes I^{A_{1}}\right)\right].\label{gammachoi}
\end{equation}
Let us substitute Eq.~\eqref{gammachoi} into Eq.~\eqref{Nforminter}.
Note that the identity channel does not change the systems it acts
on. Therefore, to express Eq.~\eqref{Nforminter} in the Choi form,
we only consider how $\phi_{+}^{A_{0}\widetilde{A}_{0}}\otimes\varphi^{E_{0}B_{0}}$
is acted on by the map $\Gamma^{\widetilde{A}_{0}E_{0}\rightarrow A_{1}}$,
representing the action of the identity channel with the identity
matrix $I^{A_{0}B_{0}}$. Thus Eq.~\eqref{Nforminter} becomes
\[
M^{AB_{0}}=\mathrm{Tr}_{\widetilde{A}_{0}E_{0}}\left[\left(I^{A_{0}B_{0}}\otimes J_{\Gamma}^{\widetilde{A}_{0}E_{0}A_{1}}\right)\left(\left(\phi_{+}^{A_{0}\widetilde{A}_{0}}\right)^{\mathrm{T}_{\widetilde{A}_{0}}}\otimes\left(\varphi^{E_{0}B_{0}}\right)^{\mathrm{T}_{E_{0}}}\otimes I^{A_{1}}\right)\right].
\]
Expanding $\phi_{+}^{A_{0}\widetilde{A}_{0}}$, and using the cyclic
property of the trace, we get:
\[
\mathrm{Tr}_{\widetilde{A}_{0}E_{0}}\left[\left(I^{A_{0}B_{0}}\otimes J_{\Gamma}^{\widetilde{A}_{0}E_{0}A_{1}}\right)\left(\left(\phi_{+}^{A_{0}\widetilde{A}_{0}}\right)^{\mathrm{T}_{\widetilde{A}_{0}}}\otimes\left(\varphi^{E_{0}B_{0}}\right)^{\mathrm{T}_{E_{0}}}\otimes I^{A_{1}}\right)\right]=
\]
\[
=\sum_{x,y}\left|x\right\rangle \left\langle y\right|^{A_{0}}\otimes\mathrm{Tr}_{E_{0}}\left[\left(I^{B_{0}}\otimes\left\langle x\middle|J_{\Gamma}^{\widetilde{A}_{0}E_{0}A_{1}}\middle|y\right\rangle ^{\widetilde{A}_{0}}\right)\left(\left(\varphi^{E_{0}B_{0}}\right)^{\mathrm{T}_{E_{0}}}\otimes I^{A}\right)\right].
\]
Since $\widetilde{A}_{0}$ is a copy of $A_{0}$,
\[
\sum_{x,y}\left|x\right\rangle \left\langle y\right|^{A_{0}}\left\langle x\middle|J_{\Gamma}^{\widetilde{A}_{0}E_{0}A_{1}}\middle|y\right\rangle ^{\widetilde{A}_{0}}=:J_{\Gamma}^{E_{0}A},
\]
where we have replaced system $\widetilde{A}_{0}$ with system $A_{0}$,
and we have set $A:=A_{0}A_{1}$ as usual. Now $\Gamma$ is regarded
as a channel from $A_{0}E_{0}$ to $A_{1}$. With this in mind, we
can write:
\[
M^{AB_{0}}=\mathrm{Tr}_{E_{0}}\left[\left(J_{\Gamma}^{E_{0}A}\otimes I^{B_{0}}\right)\left(\left(\varphi^{E_{0}B_{0}}\right)^{\mathrm{T}_{E_{0}}}\otimes I^{A}\right)\right].
\]
Taking the transpose on $E_{0}$, this expression can be rewritten
as
\[
M^{AB_{0}}=\mathrm{Tr}_{E_{0}}\left[\left(\varphi^{E_{0}B_{0}}\otimes I^{A}\right)\left(\left(J_{\Gamma}^{E_{0}A}\right)^{\mathrm{T}_{E_{0}}}\otimes I^{B_{0}}\right)\right].
\]
Now rename $E_{0}$ as $R_{0}$, and define $J_{\mathcal{N}}^{R_{0}A}:=J_{\Gamma}^{R_{0}A}$,
and $\rho^{R_{0}B_{0}}:=\varphi^{R_{0}B_{0}}$. We find that $M^{AB_{0}}$
can be written in the form of Eq.~\eqref{eq: M def}.
\end{proof}
This means that, once we require $M^{AB_{0}}$ to be positive semi-definite
with marginal $M^{A_{0}B_{0}}=I^{A_{0}}\otimes\rho^{B_{0}}$, for
some density matrix $\rho^{B_{0}}$, this automatically implies that
$M^{AB_{0}}$ has the special form of Eq.~\eqref{eq: M def}. Consequently,
we can express the requirement of complete CPTNI preservation in the
Choi form as follows: $\Theta^{A\rightarrow B}$ is complete CPTNI-preserving
if and only if $\mathbf{J}_{\Theta}^{AB}\geq0$ and $\mathrm{Tr}\left[\mathbf{J}_{\Theta}^{AB_{0}}\left(M^{AB_{0}}\right)^{\mathrm{T}}\right]\leq1$
for every positive semi-definite $M^{AB_{0}}$ with marginal $M^{A_{0}B_{0}}=I^{A_{0}}\otimes\rho^{B_{0}}$,
where $\rho^{B_{0}}\in\mathfrak{D}\left(\mathcal{H}^{B_{0}}\right)$.
This is Eq.~\eqref{eq:completechoicond} in the main article. In
particular, $\Theta^{A\rightarrow B}$ is a superchannel, which is
a completely CPTP-preserving supermap if and only if $\mathrm{Tr}\left[\mathbf{J}_{\Theta}^{AB_{0}}\left(M^{AB_{0}}\right)^{\mathrm{T}}\right]=1$
for every $M^{AB_{0}}$ as above.

Note that, among these $M^{AB_{0}}$'s we can find matrices of the
form $J_{\mathcal{N}}^{A}\otimes\rho^{B_{0}}$, where $\mathcal{N}^{A}$
is a CPTP map. These matrices are those used to check the CPTNI preservation
condition (cf.\ Eq.~\eqref{eq:choicond} in the main article). Indeed,
$J_{\mathcal{N}}^{A}\otimes\rho^{B_{0}}\geq0$, and the marginal is:
\[
\mathrm{Tr}_{A_{1}}\left[J_{\mathcal{N}}^{A}\otimes\rho^{B_{0}}\right]=\mathrm{Tr}_{A_{1}}\left[J_{\mathcal{N}}^{A}\right]\otimes\rho^{B_{0}}=I^{A_{0}}\otimes\rho^{B_{0}},
\]
because $J_{\mathcal{N}}^{A}$ is the Choi matrix of a CPTP map (see
Appendix~\ref{sec:General-facts-about}). Therefore, as it must be,
we recover in the Choi picture that CPTNI preservation is not stronger
than complete CPTNI preservation. In fact, it is strictly weaker,
as shown in Appendix~\ref{sec:counterexample}.

\section{\label{sec:counterexample}A supermap that is CPTNI-preserving, but
\emph{not} completely CPTNI-preserving}

In this appendix we present the concrete counterexample of a supermap
$\Theta^{A\rightarrow B}$ that is CPTNI-preserving, but \emph{not
completely} CPTNI-preserving. In this construction we take $\left|A_{0}\right|=\left|A_{1}\right|=\left|B_{0}\right|=2$.
Consider a supermap $\Theta^{A\rightarrow B}$ that has a Choi matrix
with marginal $\mathbf{J}_{\Theta}^{AB_{0}}=I^{A_{0}}\otimes\psi_{-}^{A_{1}B_{0}}$,
where $\psi_{-}^{A_{1}B_{0}}=\left|\psi_{-}\right\rangle \left\langle \psi_{-}\right|^{A_{1}B_{0}}$,
and $\left|\psi_{-}\right\rangle ^{A_{1}B_{0}}=\frac{1}{\sqrt{2}}\left(\left|01\right\rangle ^{A_{1}B_{0}}-\left|10\right\rangle ^{A_{1}B_{0}}\right)$
is the singlet state. Given this marginal, a possible Choi matrix
of the supermap $\Theta^{A\rightarrow B}$ is $\mathbf{J}_{\Theta}^{AB}=I^{A_{0}}\otimes\psi_{-}^{A_{1}B_{0}}\otimes u^{B_{1}}$,
where $u^{B_{1}}$ is the maximally mixed state of $B_{1}$. Now we
will prove that this supermap is CPTNI-preserving, but \emph{not completely}
CPTNI-preserving.

To this end, we first show that $\mathbf{J}_{\Theta}^{AB}$ satisfies
Eq.~\eqref{eq:choicond} in the main article. If $\mathcal{N}^{A}$
is a CPTP map and $\rho^{B_{0}}$ is a density matrix, we have 
\[
\mathrm{Tr}\left[\mathbf{J}_{\Theta}^{AB_{0}}\left(J_{\mathcal{N}}^{A}\otimes\rho^{B_{0}}\right)^{\mathrm{T}}\right]=\mathrm{Tr}\left[\left(I^{A_{0}}\otimes\psi_{-}^{A_{1}B_{0}}\right)\left(J_{\mathcal{N}}^{A}\otimes\rho^{B_{0}}\right)^{\mathrm{T}}\right].
\]
Now we express $\psi_{-}^{A_{1}B_{0}}$ in terms of the super-normalized
maximally entangled state $\phi_{+}^{A_{1}B_{0}}$:
\begin{equation}
\psi_{-}^{A_{1}B_{0}}=\frac{1}{2}\left(I^{A_{1}}\otimes Y^{B_{0}}\right)\phi_{+}^{A_{1}B_{0}}\left(I^{A_{1}}\otimes Y^{B_{0}}\right),\label{eq:psi-}
\end{equation}
where $Y^{B_{0}}$ is the Pauli $Y$ matrix. Then:
\begin{eqnarray*}
\mathrm{Tr}\left[\mathbf{J}_{\Theta}^{AB_{0}}\left(J_{\mathcal{N}}^{A}\otimes\rho^{B_{0}}\right)^{\mathrm{T}}\right] & = & \frac{1}{2}\mathrm{Tr}\left[\left(I^{A_{0}}\otimes\left(I^{A_{1}}\otimes Y^{B_{0}}\right)\phi_{+}^{A_{1}B_{0}}\left(I^{A_{1}}\otimes Y^{B_{0}}\right)\right)\left(J_{\mathcal{N}}^{A}\otimes\rho^{B_{0}}\right)^{\mathrm{T}}\right]\\
 & = & \frac{1}{2}\mathrm{Tr}_{A_{1}B_{0}}\left\{ \mathrm{Tr}_{A_{0}}\left[\left(I^{A_{0}}\otimes\left(I^{A_{1}}\otimes Y^{B_{0}}\right)\phi_{+}^{A_{1}B_{0}}\left(I^{A_{1}}\otimes Y^{B_{0}}\right)\right)\left(J_{\mathcal{N}}^{A}\otimes\rho^{B_{0}}\right)^{\mathrm{T}}\right]\right\} \\
 & = & \frac{1}{2}\mathrm{Tr}_{A_{1}B_{0}}\left[\phi_{+}^{A_{1}B_{0}}\left(\left(J_{\mathcal{N}}^{A_{1}}\right)^{\mathrm{T}}\otimes Y^{B_{0}}\left(\rho^{B_{0}}\right)^{\mathrm{T}}Y^{B_{0}}\right)\right],
\end{eqnarray*}
using the cyclic property of the trace. Now let us expand $\phi_{+}^{A_{1}B_{0}}$.
\begin{eqnarray}
\frac{1}{2}\mathrm{Tr}\left[\phi_{+}^{A_{1}B_{0}}\left(\left(J_{\mathcal{N}}^{A_{1}}\right)^{\mathrm{T}}\otimes Y^{B_{0}}\left(\rho^{B_{0}}\right)^{\mathrm{T}}Y^{B_{0}}\right)\right] & = & \frac{1}{2}\mathrm{Tr}{\textstyle \left[\sum_{x,y=1}^{2}\left|xx\right\rangle \left\langle yy\right|^{A_{1}B_{0}}\left(\left(J_{\mathcal{N}}^{A_{1}}\right)^{\mathrm{T}}\otimes Y^{B_{0}}\left(\rho^{B_{0}}\right)^{\mathrm{T}}Y^{B_{0}}\right)\right]}\nonumber \\
 & = & \frac{1}{2}\mathrm{Tr}{\textstyle \left[\sum_{x,y=1}^{2}\left\langle y\middle|\left(J_{\mathcal{N}}^{A_{1}}\right)^{\mathrm{T}}\middle|x\right\rangle ^{A_{1}}\left|x\right\rangle \left\langle y\right|^{B_{0}}Y^{B_{0}}\left(\rho^{B_{0}}\right)^{\mathrm{T}}Y^{B_{0}}\right]}\nonumber \\
 & = & \frac{1}{2}\mathrm{Tr}{\textstyle \left[\sum_{x,y=1}^{2}\left\langle x\middle|J_{\mathcal{N}}^{A_{1}}\middle|y\right\rangle ^{A_{1}}\left|x\right\rangle \left\langle y\right|^{B_{0}}Y^{B_{0}}\left(\rho^{B_{0}}\right)^{\mathrm{T}}Y^{B_{0}}\right]}.\label{eq:trick}
\end{eqnarray}
Here the expression $\sum_{x,y=1}^{2}\left\langle x\middle|J_{\mathcal{N}}^{A_{1}}\middle|y\right\rangle ^{A_{1}}\left|x\right\rangle \left\langle y\right|^{B_{0}}$
means considering $\mathcal{N}^{A}$ with its output system transformed
from $A_{1}$ to $B_{0}$. With this simplification, Eq.~\eqref{eq:trick}
reads
\[
\mathrm{Tr}\left[\mathbf{J}_{\Theta}^{AB_{0}}\left(J_{\mathcal{N}}^{A}\otimes\rho^{B_{0}}\right)^{\mathrm{T}}\right]=\frac{1}{2}\mathrm{Tr}\left[J_{\mathcal{N}}^{B_{0}}Y^{B_{0}}\left(\rho^{B_{0}}\right)^{\mathrm{T}}Y^{B_{0}}\right].
\]
Now, both $\frac{1}{2}J_{\mathcal{N}}^{B_{0}}$ and $Y^{B_{0}}\left(\rho^{B_{0}}\right)^{\mathrm{T}}Y^{B_{0}}$
are density operators, therefore
\[
\mathrm{Tr}\left[\mathbf{J}_{\Theta}^{AB_{0}}\left(J_{\mathcal{N}}^{A}\otimes\rho^{B_{0}}\right)^{\mathrm{T}}\right]=\mathrm{Tr}\left[\left(\frac{1}{2}J_{\mathcal{N}}^{B_{0}}\right)\left(Y^{B_{0}}\left(\rho^{B_{0}}\right)^{\mathrm{T}}Y^{B_{0}}\right)\right]\leq1.
\]
Hence $\mathbf{J}_{\Theta}^{AB_{0}}$ satisfies Eq.~\eqref{eq:choicond}
in the main article; therefore $\Theta^{A\rightarrow B}$ is a CPTNI-preserving
supermap.

Now we show that $\Theta^{A\rightarrow B}$ violates Eq.~\eqref{eq:completechoicond}
in the main article. To this end, let us take $M^{AB_{0}}=\left(\mathbf{J}_{\Theta}^{AB_{0}}\right)^{\mathrm{T}}$.
This choice of $M^{AB_{0}}$ complies with the two requests on $M^{AB_{0}}$
in Eq.~\eqref{eq:completechoicond} in the main article. Since $\mathbf{J}_{\Theta}^{AB_{0}}=I^{A_{0}}\otimes\psi_{-}^{A_{1}B_{0}}$,
$\left(\mathbf{J}_{\Theta}^{AB_{0}}\right)^{\mathrm{T}}$ is positive
semi-definite; and its marginal
\[
M^{A_{0}B_{0}}=\mathrm{Tr}_{A_{1}}\left[\left(\mathbf{J}_{\Theta}^{AB_{0}}\right)^{\mathrm{T}}\right]=\left(\mathrm{Tr}_{A_{1}}\left[\mathbf{J}_{\Theta}^{AB_{0}}\right]\right)^{\mathrm{T}}=\left(I^{A_{0}}\otimes u^{B_{0}}\right)^{\mathrm{T}}=I^{A_{0}}\otimes u^{B_{0}}
\]
is of the form $I^{A_{0}}\otimes\rho^{B_{0}}$, with $\rho^{B_{0}}$
density matrix. Then:
\begin{eqnarray*}
\mathrm{Tr}\left[\mathbf{J}_{\Theta}^{AB_{0}}\left(M^{AB_{0}}\right)^{\mathrm{T}}\right] & = & \mathrm{Tr}\left[\left(I^{A_{0}}\otimes\psi_{-}^{A_{1}B_{0}}\right)\left(I^{A_{0}}\otimes\psi_{-}^{A_{1}B_{0}}\right)\right]\\
 & = & \mathrm{Tr}\left[I^{A_{0}}\otimes\left(\psi_{-}^{A_{1}B_{0}}\right)^{2}\right]\\
 & = & \mathrm{Tr}_{A_{1}B_{0}}\left\{ \mathrm{Tr}_{A_{0}}\left[I^{A_{0}}\otimes\psi_{-}^{A_{1}B_{0}}\right]\right\} \\
 & = & 2\mathrm{Tr}\left[\psi_{-}^{A_{1}B_{0}}\right]\\
 & = & 2>1.
\end{eqnarray*}
This is in contrast with Eq.~\eqref{eq:completechoicond} in the
main article, therefore the supermap $\Theta^{A\rightarrow B}$\emph{
}is \emph{not} a \emph{completely} CPTNI-preserving supermap, despite
being CPTNI-preserving.

We conclude this appendix by reconstructing $\Theta^{A\rightarrow B}$
from its Choi matrix $\mathbf{J}_{\Theta}^{AB}=I^{A_{0}}\otimes\psi_{-}^{A_{1}B_{0}}\otimes u^{B_{1}}$.
By Eqs.~\eqref{eq:Choi on states} and \eqref{eq:Choi on maps},
we have, if $\mathcal{E}^{A}$ is a generic CP map, 
\begin{eqnarray*}
\Theta^{A\rightarrow B}\left[\mathcal{E}^{A}\right]\left(\rho^{B_{0}}\right) & = & \mathrm{Tr}_{AB_{0}}\left[\mathbf{J}_{\Theta}^{AB}\left(J_{\mathcal{E}}^{A}\otimes\rho^{B_{0}}\otimes I^{B_{1}}\right)^{\mathrm{T}}\right]\\
 & = & \mathrm{Tr}_{AB_{0}}\left[\left(I^{A_{0}}\otimes\psi_{-}^{A_{1}B_{0}}\right)\left(J_{\mathcal{E}}^{A}\otimes\rho^{B_{0}}\right)^{\mathrm{T}}\right]u^{B_{1}}\\
 & = & \mathrm{Tr}_{A_{1}B_{0}}\left\{ \mathrm{Tr}_{A_{0}}\left[\left(I^{A_{0}}\otimes\psi_{-}^{A_{1}B_{0}}\right)\left(J_{\mathcal{E}}^{A}\otimes\rho^{B_{0}}\right)^{\mathrm{T}}\right]\right\} u^{B_{1}}\\
 & = & \mathrm{Tr}\left[\psi_{-}^{A_{1}B_{0}}\left(J_{\mathcal{E}}^{A_{1}}\otimes\rho^{B_{0}}\right)^{\mathrm{T}}\right]u^{B_{1}}.
\end{eqnarray*}
Recalling Eq.~\eqref{eq:psi-}, we get
\[
\Theta^{A\rightarrow B}\left[\mathcal{E}^{A}\right]\left(\rho^{B_{0}}\right)=\frac{1}{2}\mathrm{Tr}\left[\phi_{+}^{A_{1}B_{0}}\left(\left(J_{\mathcal{E}}^{A_{1}}\right)^{\mathrm{T}}\otimes Y^{B_{0}}\left(\rho^{B_{0}}\right)^{\mathrm{T}}Y^{B_{0}}\right)\right]u^{B_{1}},
\]
and using an argument similar to the one in Eq.~\eqref{eq:trick},
we finally obtain
\begin{equation}
\Theta^{A\rightarrow B}\left[\mathcal{E}^{A}\right]\left(\rho^{B_{0}}\right)=\frac{1}{2}\mathrm{Tr}\left[J_{\mathcal{E}}^{B_{0}}Y^{B_{0}}\left(\rho^{B_{0}}\right)^{\mathrm{T}}Y^{B_{0}}\right]u^{B_{1}}.\label{eq:choiconv3}
\end{equation}
By Eq.~\eqref{eq:Choi on states}, $\mathcal{E}^{A_{0}\rightarrow B_{0}}\left(u^{A_{0}}\right)=\frac{1}{2}\mathrm{Tr}_{A_{0}}\left[J_{\mathcal{E}}^{A_{0}B_{0}}I^{A_{0}B_{0}}\right]=\frac{1}{2}J_{\mathcal{E}}^{B_{0}}$.
An equivalent form of Eq.~\eqref{eq:choiconv3} is, therefore,
\[
\Theta^{A\rightarrow B}\left[\mathcal{E}^{A}\right]\left(\rho^{B_{0}}\right)=\mathrm{Tr}\left[\mathcal{E}^{A_{0}\rightarrow B_{0}}\left(u^{A_{0}}\right)Y^{B_{0}}\left(\rho^{B_{0}}\right)^{\mathrm{T}}Y^{B_{0}}\right]u^{B_{1}}.
\]
This is exactly Eq.~\eqref{eq:example} in the main article.

\section{The main result\label{sec:Completion-of-Supermaps}}

In this appendix we prove the main result of this article, namely
that a supermap can be part of a super-instrument if and only if it
is completely CPTNI-preserving. To this end, it is useful to consider
the SDP~\eqref{eq:completeSDPprimal} in the main article, reported
here for the reader's convenience.
\begin{eqnarray*}
\textrm{Find} & \quad & \alpha=\max_{M}\mathrm{Tr}\left[\mathbf{J}_{\Theta}^{AB_{0}}\left(M^{AB_{0}}\right)^{\mathrm{T}}\right]\\
\textrm{Subject to:} & \quad & M^{AB_{0}}\geq0\\
 & \quad & M^{A_{0}B_{0}}=I^{A_{0}}\otimes\rho^{B_{0}}.
\end{eqnarray*}

\begin{thm}
Suppose $\Theta^{A\rightarrow B}$ is CPTNI-preserving supermap. Then
there exists another CPTNI-preserving supermap $\Theta'^{A\rightarrow B}$
such that $\Theta^{A\rightarrow B}+\Theta'^{A\rightarrow B}$ is a
superchannel if and only if $\Theta^{A\rightarrow B}$ is \emph{completely}
CPTNI-preserving.
\end{thm}

\begin{proof}
First we will show sufficiency, namely that any completely CPTNI-preserving
supermap $\Theta^{A\rightarrow B}$ can be completed to a superchannel.
Following \citep{Barvinok}, let us write the SDP~\eqref{eq:completeSDPprimal}
in the main article in a different form. To do so, consider the linear
map $\mathcal{L}:\mathfrak{B}_{h}\left(\mathcal{H}^{AB_{0}}\right)\rightarrow\mathbb{R}\oplus\mathfrak{B}_{h}\left(\mathcal{H}^{A_{0}B_{0}}\right)$,
defined as
\[
\mathcal{L}\left(M^{AB_{0}}\right)=\left(\mathrm{Tr}\left[M^{AB_{0}}\right],M^{A_{0}B_{0}}-u^{A_{0}}\otimes M^{B_{0}}\right),
\]
for every hermitian matrix $M^{AB_{0}}$. We are working with with
positive semi-definite matrices $M^{AB_{0}}$ with marginal $M^{A_{0}B_{0}}=I^{A_{0}}\otimes\rho^{B_{0}}$,
where $\rho^{B_{0}}\in\mathfrak{D}\left(\mathcal{H}^{B_{0}}\right)$,
whence
\[
\mathrm{Tr}\left[M^{AB_{0}}\right]=\mathrm{Tr}_{A_{0}B_{0}}\left\{ \mathrm{Tr}_{A_{1}}\left[M^{AB_{0}}\right]\right\} =\mathrm{Tr}\left[M^{A_{0}B_{0}}\right]=\mathrm{Tr}\left[I^{A_{0}}\otimes\rho^{B_{0}}\right]=\left|A_{0}\right|.
\]
In addition,
\[
M^{B_{0}}=\mathrm{Tr}_{A}\left[M^{AB_{0}}\right]=\mathrm{Tr}_{A_{0}}\left[M^{A_{0}B_{0}}\right]=\mathrm{Tr}_{A_{0}}\left[I^{A_{0}}\otimes\rho^{B_{0}}\right]=\left|A_{0}\right|\rho^{B_{0}}.
\]
Using $\mathcal{L}$, we can replace the condition $M^{A_{0}B_{0}}=I^{A_{0}}\otimes\rho^{B_{0}}$
with $\mathcal{L}\left(M^{AB_{0}}\right)-\left(\left|A_{0}\right|,0^{A_{0}B_{0}}\right)=\left(0,0^{A_{0}B_{0}}\right)$.
Rewriting the SDP~\eqref{eq:completeSDPprimal} in the main article
in terms of $\mathcal{L}$, one obtains:
\begin{eqnarray*}
\textrm{Find} & \quad & \alpha=\max_{M}\mathrm{Tr}\left[\mathbf{J}_{\Theta}^{AB_{0}}\left(M^{AB_{0}}\right)^{\mathrm{T}}\right]\\
\textrm{Subject to:} & \quad & \mathcal{L}\left(M^{AB_{0}}\right)-\left(\left|A_{0}\right|,0^{A_{0}B_{0}}\right)=\left(0,0^{A_{0}B_{0}}\right)\\
 & \quad & M^{AB_{0}}\ge0.
\end{eqnarray*}
We can now construct the associated dual problem as follows. The dual
map of $\mathcal{L}$ is $\mathcal{L}^{*}:\mathbb{R}\oplus\mathfrak{B}_{h}\left(\mathcal{H}^{A_{0}B_{0}}\right)\rightarrow\mathfrak{B}_{h}\left(\mathcal{H}^{AB_{0}}\right)$
such that
\[
\mathcal{L}^{*}\left(r,\sigma^{A_{0}B_{0}}\right)=\left(rI^{A_{0}B_{0}}+\sigma^{A_{0}B_{0}}-u^{A_{0}}\otimes\sigma^{B_{0}}\right)\otimes I^{A_{1}},
\]
where $\left(r,\sigma^{A_{0}B_{0}}\right)\in\mathbb{R}\oplus\mathfrak{B}_{h}\left(\mathcal{H}^{A_{0}B_{0}}\right)$.
The dual problem is then:
\begin{eqnarray*}
\textrm{Find} & \quad & \beta=\min\left\langle \left(r,\sigma^{A_{0}B_{0}}\right),\left(\left|A_{0}\right|,0\right)\right\rangle \\
\textrm{Subject to:} & \quad & \left(rI^{A_{0}B_{0}}+\sigma^{A_{0}B_{0}}-u^{A_{0}}\otimes\sigma^{B_{0}}\right)\otimes I^{A_{1}}-\mathbf{J}_{\Theta}^{AB_{0}}\geq0\\
 & \quad & r\in\mathbb{R}\\
 & \quad & \sigma^{A_{0}B_{0}}\in\mathfrak{B}_{h}\left(\mathcal{H}^{A_{0}B_{0}}\right),
\end{eqnarray*}
where the inner product $\left\langle \left(r,\sigma^{A_{0}B_{0}}\right),\left(s,\tau^{A_{0}B_{0}}\right)\right\rangle $
is defined as
\[
\left\langle \left(r,\sigma^{A_{0}B_{0}}\right),\left(s,\tau^{A_{0}B_{0}}\right)\right\rangle =rs+\mathrm{Tr}\left[\sigma^{A_{0}B_{0}}\tau^{A_{0}B_{0}}\right].
\]
With this in mind, the dual problem simplifies to:
\begin{eqnarray}
\textrm{Find} & \quad & \beta=\left|A_{0}\right|\min r\nonumber \\
\textrm{Subject to:} & \quad & \left(rI^{A_{0}B_{0}}+\sigma^{A_{0}B_{0}}-u^{A_{0}}\otimes\sigma^{B_{0}}\right)\otimes I^{A_{1}}\geq\mathbf{J}_{\Theta}^{AB_{0}}\label{eq:Dual}\\
 & \quad & r\in\mathbb{R}\nonumber \\
 & \quad & \sigma^{A_{0}B_{0}}\in\mathfrak{B}_{h}\left(\mathcal{H}^{A_{0}B_{0}}\right).\nonumber 
\end{eqnarray}
Notice that the matrix $rI^{A_{0}B_{0}}+\sigma^{A_{0}B_{0}}-u^{A_{0}}\otimes\sigma^{B_{0}}$
must be positive semi-definite, otherwise the first constraint could
not be satisfied. In particular this implies $r\geq0$. Indeed, if
$r<0$, for some $\sigma^{A_{0}B_{0}}$ the matrix $rI^{A_{0}B_{0}}+\sigma^{A_{0}B_{0}}-u^{A_{0}}\otimes\sigma^{B_{0}}$
would have negative eigenvalues. Factoring $r\left|A_{0}\right|$
out of the first term of the constraint in Eq.~\eqref{eq:Dual},
we get
\[
\left(rI^{A_{0}B_{0}}+\sigma^{A_{0}B_{0}}-u^{A_{0}}\otimes\sigma^{B_{0}}\right)\otimes I^{A_{1}}=r\left|A_{0}\right|\left(u^{A_{0}}\otimes I^{B_{0}}+\sigma'^{A_{0}B_{0}}-u^{A_{0}}\otimes\sigma'^{B_{0}}\right)\otimes I^{A_{1}},
\]
where $\sigma'^{A_{0}B_{0}}:=\frac{1}{r\left|A_{0}\right|}\sigma^{A_{0}B_{0}}$
if $r\neq0$. Note that this does not alter the constraint on the
dual SDP, so we can forget the primes, and rewrite Eq.~\eqref{eq:Dual}
as:
\begin{eqnarray*}
\textrm{Find} & \quad & \beta=\left|A_{0}\right|\min r\\
\textrm{Subject to:} & \quad & r\left|A_{0}\right|\left(u^{A_{0}}\otimes I^{B_{0}}+\sigma^{A_{0}B_{0}}-u^{A_{0}}\otimes\sigma^{B_{0}}\right)\otimes I^{A_{1}}\geq\mathbf{J}_{\Theta}^{AB_{0}}\\
 & \quad & r\geq0\\
 & \quad & \sigma^{A_{0}B_{0}}\in\mathfrak{B}_{h}\left(\mathcal{H}^{A_{0}B_{0}}\right).
\end{eqnarray*}
In particular, this implies that $u^{A_{0}}\otimes I^{B_{0}}+\sigma^{A_{0}B_{0}}-u^{A_{0}}\otimes\sigma^{B_{0}}\geq0$.
Now let us define
\begin{equation}
\mathbf{J}_{\Phi}^{AB_{0}}:=\left(u^{A_{0}}\otimes I^{B_{0}}+\sigma^{A_{0}B_{0}}-u^{A_{0}}\otimes\sigma^{B_{0}}\right)\otimes I^{A_{1}}.\label{eq:=00005CJPhi}
\end{equation}
Note that $\mathbf{J}_{\Phi}^{AB_{0}}=\mathbf{J}_{\Phi}^{A_{0}B_{0}}\otimes u^{A_{1}}$
because
\[
\mathbf{J}_{\Phi}^{A_{0}B_{0}}=\mathrm{Tr}_{A_{1}}\left[\mathbf{J}_{\Phi}^{AB_{0}}\right]=\left|A_{1}\right|\left(u^{A_{0}}\otimes I^{B_{0}}+\sigma^{A_{0}B_{0}}-u^{A_{0}}\otimes\sigma^{B_{0}}\right).
\]
Moreover,
\[
\mathbf{J}_{\Phi}^{A_{1}B_{0}}=\mathrm{Tr}_{A_{0}}\left[\mathbf{J}_{\Phi}^{AB_{0}}\right]=\left(I^{B_{0}}+\sigma^{B_{0}}-\sigma^{B_{0}}\right)\otimes I^{A_{1}}=I^{A_{1}B_{0}}.
\]
Since $\mathbf{J}_{\Phi}^{A_{0}B_{0}}\geq0$, by Appendix~\ref{subsec:Useful-results-about}
$\mathbf{J}_{\Phi}^{AB_{0}}$ is the marginal Choi matrix of a superchannel
$\Phi^{A\rightarrow B}$. Eq.~\eqref{eq:=00005CJPhi} can be taken
as the definition of the marginal $\mathbf{J}_{\Phi}^{AB_{0}}$ of
the Choi matrix of \emph{any} superchannel. This is because any such
marginal $\mathbf{J}_{\Phi}^{AB_{0}}$ can be written as in Eq.~\eqref{eq:=00005CJPhi}
for some hermitian matrix $\sigma^{A_{0}B_{0}}$: it is enough to
take $\sigma^{A_{0}B_{0}}$ to be $\frac{1}{\left|A_{1}\right|}\mathbf{J}_{\Phi}^{A_{0}B_{0}}$.
Indeed, substituting $\sigma^{A_{0}B_{0}}=\frac{1}{\left|A_{1}\right|}\mathbf{J}_{\Phi}^{A_{0}B_{0}}$
in the right-hand side of Eq.~\eqref{eq:=00005CJPhi} yields:
\begin{align*}
\left|A_{1}\right|\left(u^{A_{0}}\otimes I^{B_{0}}+\frac{1}{\left|A_{1}\right|}\mathbf{J}_{\Phi}^{A_{0}B_{0}}-\frac{1}{\left|A_{1}\right|}u^{A_{0}}\otimes\mathbf{J}_{\Phi}^{B_{0}}\right)\otimes u^{A_{1}} & =\left(\left|A_{1}\right|u^{A_{0}}\otimes I^{B_{0}}+\mathbf{J}_{\Phi}^{A_{0}B_{0}}-u^{A_{0}}\otimes\mathbf{J}_{\Phi}^{B_{0}}\right)\otimes u^{A_{1}}\\
 & =\left(\left|A_{1}\right|u^{A_{0}}\otimes I^{B_{0}}+\mathbf{J}_{\Phi}^{A_{0}B_{0}}-u^{A_{0}}\otimes\mathrm{Tr}_{AB_{1}}\left[\mathbf{J}_{\Phi}^{AB}\right]\right)\otimes u^{A_{1}}\\
 & =\left(\left|A_{1}\right|u^{A_{0}}\otimes I^{B_{0}}+\mathbf{J}_{\Phi}^{A_{0}B_{0}}-u^{A_{0}}\otimes\mathrm{Tr}_{A_{1}}\left\{ \mathrm{Tr}_{A_{0}B_{1}}\left[\mathbf{J}_{\Phi}^{AB}\right]\right\} \right)\otimes u^{A_{1}}\\
 & =\left(\left|A_{1}\right|u^{A_{0}}\otimes I^{B_{0}}+\mathbf{J}_{\Phi}^{A_{0}B_{0}}-u^{A_{0}}\otimes\mathrm{Tr}_{A_{1}}\left[\mathbf{J}_{\Phi}^{A_{1}B_{0}}\right]\right)\otimes u^{A_{1}}\\
 & =\left(\left|A_{1}\right|u^{A_{0}}\otimes I^{B_{0}}+\mathbf{J}_{\Phi}^{A_{0}B_{0}}-u^{A_{0}}\otimes\mathrm{Tr}_{A_{1}}\left[I^{A_{1}B_{0}}\right]\right)\otimes u^{A_{1}}\\
 & =\left(\left|A_{1}\right|u^{A_{0}}\otimes I^{B_{0}}+\mathbf{J}_{\Phi}^{A_{0}B_{0}}-\left|A_{1}\right|u^{A_{0}}\otimes I^{B_{0}}\right)\otimes u^{A_{1}}\\
 & =\mathbf{J}_{\Phi}^{A_{0}B_{0}}\otimes u^{A_{1}}.
\end{align*}
Therefore, in the light of these remarks, the dual SDP can equivalently
be formulated in the following terms:
\begin{eqnarray*}
\textrm{Find} & \quad & \beta=\left|A_{0}\right|\min r\\
\textrm{Subject to:} & \quad & r\left|A_{0}\right|\mathbf{J}_{\Phi}^{A_{0}B_{0}}\otimes u^{A_{1}}\geq\mathbf{J}_{\Theta}^{AB_{0}}\\
 & \quad & \mathbf{J}_{\Phi}^{A_{0}B_{0}}\geq0\\
 & \quad & \mathbf{J}_{\Phi}^{A_{1}B_{0}}=I^{A_{1}B_{0}}\\
 & \quad & r\geq0.
\end{eqnarray*}
Strong duality states that the primal and dual problem have the same
optimal solution, therefore $\alpha=\beta$. Since $\Theta^{A\rightarrow B}$
is completely CPTNI-preserving, $\alpha=\max_{M}\mathrm{Tr}\left[\mathbf{J}_{\Theta}^{AB_{0}}\left(M^{AB_{0}}\right)^{\mathrm{T}}\right]\leq1$.
Hence $\beta\leq1$. Clearly taking $r\left|A_{0}\right|=\beta$ satisfies
the constraint $r\left|A_{0}\right|\mathbf{J}_{\Phi}^{A_{0}B_{0}}\otimes u^{A_{1}}\geq\mathbf{J}_{\Theta}^{AB_{0}}$,
and we have
\[
\mathbf{J}_{\Phi}^{A_{0}B_{0}}\otimes u^{A_{1}}\geq\beta\mathbf{J}_{\Phi}^{A_{0}B_{0}}\otimes u^{A_{1}}\geq\mathbf{J}_{\Theta}^{AB_{0}},
\]
because $\beta\leq1$. Now define $\Theta'^{A\rightarrow B}$ to be
a new supermap such that $\mathbf{J}_{\Theta'}^{AB_{0}}:=\mathbf{J}_{\Phi}^{A_{0}B_{0}}\otimes u^{A_{1}}-\mathbf{J}_{\Theta}^{AB_{0}}$.
By construction $\mathbf{J}_{\Theta'}^{AB_{0}}\geq0$; and by substituting
$\mathbf{J}_{\Theta'}^{AB_{0}}$ into the left-hand side of Eq.~\eqref{eq:choicond}
in the main article one obtains:
\begin{eqnarray*}
\mathrm{Tr}\left[\mathbf{J}_{\Theta'}^{AB_{0}}\left(J_{\mathcal{N}}^{A}\otimes\rho^{B_{0}}\right)^{\mathrm{T}}\right] & = & \mathrm{Tr}\left[\left(\mathbf{J}_{\Phi}^{AB_{0}}-\mathbf{J}_{\Theta}^{AB_{0}}\right)\left(J_{\mathcal{N}}^{A}\otimes\rho^{B_{0}}\right)^{\mathrm{T}}\right]\\
 & = & \mathrm{Tr}\left[\mathbf{J}_{\Phi}^{AB_{0}}\left(J_{\mathcal{N}}^{A}\otimes\rho^{B_{0}}\right)^{\mathrm{T}}\right]-\mathrm{Tr}\left[\mathbf{J}_{\Theta}^{AB_{0}}\left(J_{\mathcal{N}}^{A}\otimes\rho^{B_{0}}\right)^{\mathrm{T}}\right]\\
 & = & 1-\mathrm{Tr}\left[\mathbf{J}_{\Theta}^{AB_{0}}\left(J_{\mathcal{N}}^{A}\otimes\rho^{B_{0}}\right)^{\mathrm{T}}\right],
\end{eqnarray*}
where we have used the fact that $\Phi^{A\rightarrow B}$ is a superchannel
(see Appendix~\ref{subsec:Useful-results-about}). Now, $\mathrm{Tr}\left[\mathbf{J}_{\Theta}^{AB_{0}}\left(J_{\mathcal{N}}^{A}\otimes\rho^{B_{0}}\right)^{\mathrm{T}}\right]\geq0$
because $\Theta^{A\rightarrow B}$ is CPP. Therefore $\mathrm{Tr}\left[\mathbf{J}_{\Theta'}^{AB_{0}}\left(J_{\mathcal{N}}^{A}\otimes\rho^{B_{0}}\right)^{\mathrm{T}}\right]\leq1$
for every $J_{\mathcal{N}}^{A}\otimes\rho^{B_{0}}$, thus $\Theta'^{A\rightarrow B}$
is CPTNI-preserving.

To conclude the proof, let us prove necessity. Assume that $\Theta^{A\rightarrow B}$
is a CPTNI-preserving supermap such that $\Phi^{A\rightarrow B}=\Theta^{A\rightarrow B}+\Theta'^{A\rightarrow B}$
is a superchannel, where $\Theta'^{A\rightarrow B}$ is another CPTNI-preserving
supermap. We will prove that $\Theta^{A\rightarrow B}$ must be \emph{completely}
CPTNI-preserving. In the Choi picture we have
\begin{equation}
\mathbf{J}_{\Theta}^{AB_{0}}+\mathbf{J}_{\Theta'}^{AB_{0}}=\mathbf{J}_{\Phi}^{AB_{0}}.\label{eq:choi}
\end{equation}
Let us multiply both sides of Eq.~\eqref{eq:choi} by the transpose
of any matrix $M^{AB_{0}}\geq0$ with marginal $M^{A_{0}B_{0}}=I^{A_{0}}\otimes\rho^{B_{0}}$,
$\rho^{B_{0}}\in\mathfrak{D}\left(\mathcal{H}^{B_{0}}\right)$, and
then take the trace.
\begin{equation}
\mathrm{Tr}\left[\mathbf{J}_{\Theta}^{AB_{0}}\left(M^{AB_{0}}\right)^{\mathrm{T}}\right]+\mathrm{Tr}\left[\mathbf{J}_{\Theta'}^{AB_{0}}\left(M^{AB_{0}}\right)^{\mathrm{T}}\right]=\mathrm{Tr}\left[\mathbf{J}_{\Phi}^{AB_{0}}\left(M^{AB_{0}}\right)^{\mathrm{T}}\right]\label{eq:choiM}
\end{equation}
By the results in Appendix~\ref{subsec:Some-technical-derivations},
the right-hand side is 1 because $\Phi^{A\rightarrow B}$ is a superchannel.
Thus Eq.~\eqref{eq:choiM} becomes
\[
\mathrm{Tr}\left[\mathbf{J}_{\Theta}^{AB_{0}}\left(M^{AB_{0}}\right)^{\mathrm{T}}\right]+\mathrm{Tr}\left[\mathbf{J}_{\Theta'}^{AB_{0}}\left(M^{AB_{0}}\right)^{\mathrm{T}}\right]=1,
\]
which implies $\mathrm{Tr}\left[\mathbf{J}_{\Theta}^{AB_{0}}\left(M^{AB_{0}}\right)^{\mathrm{T}}\right]\leq1$
for all $M^{AB_{0}}$ because $\Theta^{A\rightarrow B}$ is CPP. Therefore
$\Theta^{A\rightarrow B}$ satisfies Eq.~\eqref{eq:completechoicond}
in the main article, which means that it is completely CPTNI-preserving.
This concludes the proof.
\end{proof}
Applying the statement of this theorem to $\Theta'^{A\rightarrow B}$,
we get that $\Theta'^{A\rightarrow B}$ is completely CPTNI-preserving
too.

\section{Quantum super-instruments\label{sec:Quantum-Super-measurements}}

In this appendix we re-derive one of the results of \citep{Chiribella2008},
but in a different way. This new proof is based on our main result:
every completely CPTNI-preserving supermap can be completed to a superchannel.
Specifically, we show that each completely CPTNI-preserving supermap
$\Theta_{x}^{A\rightarrow B}$ in a super-measurement $\left\{ \Theta_{x}^{A\rightarrow B}\right\} $
can be expressed in terms of a CPTP pre-processing channel, independent
of $x$, and a CPTNI post-processing map, as depicted in Fig.~\ref{fig:superchannel}.
\begin{prop}
\label{prop:super-measurement}The Choi matrix $\mathbf{J}_{\Theta_{x}}^{AB}$
of each completely CPTNI-preserving supermap $\Theta_{x}^{A\rightarrow B}$
in a super-measurement $\left\{ \Theta_{x}^{A\rightarrow B}\right\} _{x\in X}$
can be written in terms of a common CPTP pre-processing map $\Gamma_{\mathrm{pre}}^{\widetilde{B}_{0}\rightarrow A_{0}E_{0}}$,
and a CPTNI post-processing map $\Gamma_{\mathrm{post}_{x}}^{\widetilde{A}_{1}E_{0}\rightarrow B_{1}}$
as
\[
\mathbf{J}_{\Theta_{x}}^{AB}=\left(\mathcal{I}^{AB_{0}}\otimes\Gamma_{\mathrm{post}_{x}}^{\widetilde{A}_{1}E_{0}\rightarrow B_{1}}\right)\circ\left(\mathcal{I}^{A_{1}\widetilde{A}_{1}B_{0}}\otimes\Gamma_{\mathrm{pre}}^{\widetilde{B}_{0}\rightarrow A_{0}E_{0}}\right)\left(\phi_{+}^{B_{0}\widetilde{B}_{0}}\otimes\phi_{+}^{A_{1}\widetilde{A}_{1}}\right).
\]
\end{prop}

\begin{proof}
Define $\Theta^{A\rightarrow B}:=\sum_{x\in X}\Theta_{x}^{A\rightarrow B}$,
which we know to be a superchannel. In the Choi picture this is translated
into $\mathbf{J}_{\Theta}^{AB}=\sum_{x}\mathbf{J}_{\Theta_{x}}^{AB}$.
In \citep[theorem 1]{Gour2018} one of the authors showed that the
Choi matrix of a superchannel can be written in terms of its pre-processing
$\Gamma_{\mathrm{pre}}^{B_{0}\rightarrow A_{0}E_{0}}$ and post-processing
$\Gamma_{\mathrm{post}}^{A_{1}E_{0}\rightarrow B_{1}}$ as
\[
\mathbf{J}_{\Theta}^{AB}=\left(\mathcal{I}^{AB_{0}}\otimes\Gamma_{\mathrm{post}}^{\widetilde{A}_{1}E_{0}\rightarrow B_{1}}\right)\left(\psi^{A_{0}B_{0}E_{0}}\otimes\phi_{+}^{A_{1}\widetilde{A}_{1}}\right),
\]
where $\psi^{A_{0}B_{0}E_{0}}:=\left(\mathcal{I}^{AB_{0}}\otimes\Gamma_{\mathrm{pre}}^{\widetilde{B}_{0}\rightarrow A_{0}E_{0}}\right)\left(\phi_{+}^{B_{0}\widetilde{B}_{0}}\right)$.
$\psi^{A_{0}B_{0}E_{0}}$ can be shown to be a purification of $\frac{1}{\left|A_{1}\right|}\mathbf{J}_{\Theta}^{A_{0}B_{0}}$
\citep{Gour2018}. Now, summing over all outcomes $x\in X$, let us
construct the matrix $\sum_{x\in X}\left|x\right\rangle \left\langle x\right|^{X_{1}}\otimes\mathbf{J}_{\Theta_{x}}^{AB}$,
where $\left\{ \left|x\right\rangle ^{X_{1}}\right\} _{x=1}^{\left|X_{1}\right|}$
is an orthonormal basis of $\mathcal{H}^{X_{1}}$. Let $\varphi^{X_{1}ABF_{0}}$
be a purification of $\sum_{x}\left|x\right\rangle \left\langle x\right|^{X_{1}}\otimes\mathbf{J}_{\Theta_{x}}^{AB}$.
Note that $\varphi^{X_{1}ABF_{0}}$ is a purification of $\mathbf{J}_{\Theta}^{AB}$
too, because
\[
\mathrm{Tr}_{X_{1}F_{0}}\left[\varphi^{X_{1}ABF_{0}}\right]=\mathrm{Tr}_{X_{1}}\left\{ \mathrm{Tr}_{F_{0}}\left[\varphi^{X_{1}ABF_{0}}\right]\right\} =\mathrm{Tr}_{X_{1}}\left[\sum_{x}\left|x\right\rangle \left\langle x\right|^{X_{1}}\otimes\mathbf{J}_{\Theta_{x}}^{AB}\right]=\sum_{x}\mathbf{J}_{\Theta_{x}}^{AB}=\mathbf{J}_{\Theta}^{AB}.
\]
If we take the isometry $\mathcal{V}^{\widetilde{A}_{1}E_{0}\rightarrow B_{1}G_{0}}$
to be a Stinespring dilation of $\Gamma_{\mathrm{post}}^{\widetilde{A}_{1}E_{0}\rightarrow B_{1}}$,
namely $\Gamma_{\mathrm{post}}^{\widetilde{A}_{1}E_{0}\rightarrow B_{1}}=\mathrm{Tr}_{G_{0}}\circ\mathcal{V}^{\widetilde{A}_{1}E_{0}\rightarrow B_{1}G_{0}}$,
then
\[
\chi^{ABG_{0}}:=\left(\mathcal{I}^{AB_{0}}\otimes\mathcal{V}^{\widetilde{A}_{1}E_{0}\rightarrow B_{1}G_{0}}\right)\left(\psi^{A_{0}B_{0}E_{0}}\otimes\phi_{+}^{A_{1}\widetilde{A}_{1}}\right)
\]
is another purification of $\mathbf{J}_{\Theta}^{AB}$. Indeed, 
\begin{align*}
\mathrm{Tr}_{G_{0}}\left[\chi^{ABG_{0}}\right] & =\mathcal{I}^{AB_{0}}\otimes\left(\mathrm{Tr}_{G_{0}}\circ\mathcal{V}^{\widetilde{A}_{1}E_{0}\rightarrow B_{1}G_{0}}\right)\left(\psi^{A_{0}B_{0}E_{0}}\otimes\phi_{+}^{A_{1}\widetilde{A}_{1}}\right)\\
 & =\left(\mathcal{I}^{AB_{0}}\otimes\Gamma_{\mathrm{post}}^{\widetilde{A}_{1}E_{0}\rightarrow B_{1}}\right)\left(\psi^{A_{0}B_{0}E_{0}}\otimes\phi_{+}^{A_{1}\widetilde{A}_{1}}\right)\\
 & =\mathbf{J}_{\Theta}^{AB}.
\end{align*}
Since both $\varphi^{X_{1}ABF_{0}}$ and $\chi^{ABG_{0}}$ are purifications
of $\mathbf{J}_{\Theta}^{AB}$, they are related by an isometry channel
$\mathcal{U}^{G_{0}\rightarrow X_{1}F_{0}}$ such that
\begin{eqnarray*}
\varphi^{X_{1}ABF_{0}} & = & \left(\mathcal{I}^{AB}\otimes\mathcal{U}^{G_{0}\rightarrow X_{1}F_{0}}\right)\left(\chi^{ABG_{0}}\right)\\
 & = & \left(\mathcal{I}^{AB}\otimes\mathcal{U}^{G_{0}\rightarrow X_{1}F_{0}}\right)\left(\mathcal{I}^{AB_{0}}\otimes\mathcal{V}^{\widetilde{A}_{1}E_{0}\rightarrow B_{1}G_{0}}\right)\left(\psi^{A_{0}B_{0}E_{0}}\otimes\phi_{+}^{A_{1}\widetilde{A}_{1}}\right)\\
 & =: & \left(\mathcal{I}^{AB_{0}}\otimes\mathcal{W}^{\widetilde{A}_{1}E_{0}\rightarrow B_{1}X_{1}F_{0}}\right)\left(\psi^{A_{0}B_{0}E_{0}}\otimes\phi_{+}^{A_{1}\widetilde{A}_{1}}\right),
\end{eqnarray*}
where we have defined $\mathcal{W}^{\widetilde{A}_{1}E_{0}\rightarrow B_{1}X_{1}F_{0}}:=\mathcal{U}^{G_{0}\rightarrow X_{1}F_{0}}\circ\mathcal{V}^{\widetilde{A}_{1}E_{0}\rightarrow B_{1}G_{0}}$,
which is another isometry channel, and another Stinespring dilation
of $\Gamma_{\mathrm{post}}^{\widetilde{A}_{1}E_{0}\rightarrow B_{1}}$.
Now let us trace out system $F_{0}$, recalling that $\mathrm{Tr}_{F_{0}}\left[\varphi^{X_{1}ABF_{0}}\right]=\sum_{y}\left|y\right\rangle \left\langle y\right|^{X_{1}}\otimes\mathbf{J}_{\Theta_{y}}^{AB}$,
where we have changed the index from $x$ to $y$ for convenience.
We get
\begin{equation}
\sum_{y}\left|y\right\rangle \left\langle y\right|^{X_{1}}\otimes\mathbf{J}_{\Theta_{y}}^{AB}=\left(\mathcal{I}^{AB_{0}}\otimes\widetilde{\Gamma}^{\widetilde{A}_{1}E_{0}\rightarrow B_{1}X_{1}}\right)\left(\psi^{A_{0}B_{0}E_{0}}\otimes\phi_{+}^{A_{1}\widetilde{A}_{1}}\right),\label{eq:Gamma tilde}
\end{equation}
where $\widetilde{\Gamma}^{\widetilde{A}_{1}E_{0}\rightarrow B_{1}X_{1}}:=\mathrm{Tr}_{F_{0}}\circ\mathcal{W}^{\widetilde{A}_{1}E_{0}\rightarrow B_{1}X_{1}F_{0}}$
is a CPTP map. To get $\mathbf{J}_{\Theta_{x}}^{AB}$, we apply the
projector $\left|x\right\rangle \left\langle x\right|^{X_{1}}$ to
both sides of Eq.~\eqref{eq:Gamma tilde}, tracing over $X_{1}$:
\begin{align}
\mathbf{J}_{\Theta_{x}}^{AB} & =\mathrm{Tr}_{X_{1}}\left[\left|x\right\rangle \left\langle x\right|^{X_{1}}\left(\mathcal{I}^{AB_{0}}\otimes\widetilde{\Gamma}^{\widetilde{A}_{1}E_{0}\rightarrow B_{1}X_{1}}\right)\left(\psi^{A_{0}B_{0}E_{0}}\otimes\phi_{+}^{A_{1}\widetilde{A}_{1}}\right)\right]\nonumber \\
 & =\left[\mathcal{I}^{AB_{0}}\otimes\left(\mathrm{Tr}_{X_{1}}\left|x\right\rangle \left\langle x\right|^{X_{1}}\circ\widetilde{\Gamma}^{\widetilde{A}_{1}E_{0}\rightarrow B_{1}X_{1}}\right)\right]\left(\psi^{A_{0}B_{0}E_{0}}\otimes\phi_{+}^{A_{1}\widetilde{A}_{1}}\right).\label{eq:=00005CJ_x}
\end{align}
Now let us define $\Gamma_{\mathrm{post}_{x}}^{\widetilde{A}_{1}E_{0}\rightarrow B_{1}}:=\mathrm{Tr}_{X_{1}}\left|x\right\rangle \left\langle x\right|^{X_{1}}\circ\widetilde{\Gamma}^{\widetilde{A}_{1}E_{0}\rightarrow B_{1}X_{1}}$,
which is a CPTNI map whose action on a density matrix $\rho^{\widetilde{A}_{1}E_{0}}$
is
\[
\Gamma_{\mathrm{post}_{x}}^{\widetilde{A}_{1}E_{0}\rightarrow B_{1}}\left(\rho^{\widetilde{A}_{1}E_{0}}\right)=\left\langle x\middle|\widetilde{\Gamma}^{\widetilde{A}_{1}E_{0}\rightarrow B_{1}X_{1}}\left(\rho^{\widetilde{A}_{1}E_{0}}\right)\middle|x\right\rangle ^{X_{1}}.
\]
Therefore Eq.~\eqref{eq:=00005CJ_x} becomes
\[
\mathbf{J}_{\Theta_{x}}^{AB}=\left(\mathcal{I}^{AB_{0}}\otimes\Gamma_{\mathrm{post}_{x}}^{\widetilde{A}_{1}E_{0}\rightarrow B_{1}}\right)\left(\psi^{A_{0}B_{0}E_{0}}\otimes\phi_{+}^{A_{1}\widetilde{A}_{1}}\right).
\]
Recalling $\psi^{A_{0}B_{0}E_{0}}=\left(\mathcal{I}^{AB_{0}}\otimes\Gamma_{\mathrm{pre}}^{\widetilde{B}_{0}\rightarrow A_{0}E_{0}}\right)\left(\phi_{+}^{B_{0}\widetilde{B}_{0}}\right)$,
where $\Gamma_{\mathrm{pre}}^{\widetilde{B}_{0}\rightarrow A_{0}E_{0}}$
is the pre-processing of the superchannel $\Theta^{A\rightarrow B}$,
we get the thesis:
\[
\mathbf{J}_{\Theta_{x}}^{AB}=\left(\mathcal{I}^{AB_{0}}\otimes\Gamma_{\mathrm{post}_{x}}^{\widetilde{A}_{1}E_{0}\rightarrow B_{1}}\right)\circ\left(\mathcal{I}^{A_{1}\widetilde{A}_{1}B_{0}}\otimes\Gamma_{\mathrm{pre}}^{\widetilde{B}_{0}\rightarrow A_{0}E_{0}}\right)\left(\phi_{+}^{B_{0}\widetilde{B}_{0}}\otimes\phi_{+}^{A_{1}\widetilde{A}_{1}}\right).
\]
\end{proof}
Therefore we can realize every completely CPTNI-preserving supermap
$\Theta_{x}^{A\rightarrow B}$ that is part of a quantum super-instrument
as a quantum 1-comb, as in Fig.~\ref{fig:superchannel}. More precisely
we have\begin{equation}\label{eq:Theta_x comb}
\Theta^{A\rightarrow B}_{x}=~\begin{aligned}\Qcircuit @C=1em @R=.7em @!R {& \qw \poloFantasmaCn{B_0}  & \multigate{1}{\Gamma_{\mathrm{pre}}} & \qw \poloFantasmaCn{A_0} &\qw & & &\qw \poloFantasmaCn{A_1} & \multigate{1}{\Gamma_{\mathrm{post}_{x}}} & \qw \poloFantasmaCn{B_1} &\qw \\ && \pureghost{\Gamma_{\mathrm{pre}}} & \qw \poloFantasmaCn{E_0}  & \qw &\qw &\qw &\qw & \ghost{\Gamma_{\mathrm{post}_{x}}}}\end{aligned}~,
\end{equation}where $\Gamma_{\mathrm{pre}}^{B_{0}\rightarrow A_{0}E_{0}}$ is the
CPTP pre-processing of the superchannel $\Theta^{A\rightarrow B}=\sum_{x}\Theta_{x}^{A\rightarrow B}$.
The pre-processing of a completely CPTNI-preserving supermap is therefore
independent of $x$ and common to all the supermaps in the same quantum
super-instrument. In fact, even the post-processing is almost shared
by all supermaps in the same super-instrument: it is given by $\Gamma_{\mathrm{post}_{x}}^{A_{1}E_{0}\rightarrow B_{1}}=\mathrm{Tr}_{X_{1}}\left|x\right\rangle \left\langle x\right|^{X_{1}}\circ\widetilde{\Gamma}^{A_{1}E_{0}\rightarrow B_{1}X_{1}}$,
namely by a reading performed on the classical output $X_{1}$ of
$\widetilde{\Gamma}^{A_{1}E_{0}\rightarrow B_{1}X_{1}}$. $\widetilde{\Gamma}^{A_{1}E_{0}\rightarrow B_{1}X_{1}}$
depends only on the superchannel $\Theta^{A\rightarrow B}$, so it
is common to all the supermaps in the same super-instrument. Eq.~\eqref{eq:Theta_x comb}
then becomes\[
\Theta^{A\rightarrow B}_{x}=~\begin{aligned}\Qcircuit @C=1em @R=.7em @!R {& \qw \poloFantasmaCn{B_0}  & \multigate{1}{\Gamma_{\mathrm{pre}}} & \qw \poloFantasmaCn{A_0} &\qw & & &\qw \poloFantasmaCn{A_1} & \multigate{1}{\widetilde{\Gamma}} & \qw \poloFantasmaCn{B_1} &\qw \\ && \pureghost{\Gamma_{\mathrm{pre}}} & \qw \poloFantasmaCn{E_0}  & \qw &\qw &\qw &\qw & \ghost{\widetilde{\Gamma}}& \qw \poloFantasmaCn{X_1} &\measureD{x}}\end{aligned}~.
\]

From the proof of proposition~\ref{prop:super-measurement} we have
\begin{equation}
\mathrm{Tr}_{X_{1}}\circ\widetilde{\Gamma}^{A_{1}E_{0}\rightarrow B_{1}X_{1}}=\mathrm{Tr}_{X_{1}F_{0}}\circ\mathcal{W}^{A_{1}E_{0}\rightarrow B_{1}X_{1}F_{0}}=\Gamma_{\mathrm{post}}^{A_{1}E_{0}\rightarrow B_{1}},\label{eq:Gamma tilde - Gamma post}
\end{equation}
because $\mathcal{W}^{A_{1}E_{0}\rightarrow B_{1}X_{1}F_{0}}$ is
a Stinespring dilation of $\Gamma_{\mathrm{post}}^{A_{1}E_{0}\rightarrow B_{1}}$.
Therefore, if we forget the outcome $x$ of the super-measurement,
Eq.~\eqref{eq:Gamma tilde - Gamma post} yields
\[
\sum_{x}\Gamma_{\mathrm{post}_{x}}^{A_{1}E_{0}\rightarrow B_{1}}=\mathrm{Tr}_{X_{1}}\circ\widetilde{\Gamma}^{A_{1}E_{0}\rightarrow B_{1}X_{1}}=\Gamma_{\mathrm{post}}^{A_{1}E_{0}\rightarrow B_{1}},
\]
and we recover the post-processing channel $\Gamma_{\mathrm{post}}^{A_{1}E_{0}\rightarrow B_{1}}$
of $\Theta^{A\rightarrow B}.$

\section{OPT interpretation of the result\label{sec:OPT-interpretation-of}}

The theory of quantum supermaps, where generic evolutions of quantum
maps are described by supermaps, can be analysed using the framework
of operational probabilistic theories (OPTs) \citep{Chiribella-purification,Chiribella-informational,Chiribella14,QuantumFromPrinciples,hardy2011,Hardy-informational-2,hardy2013},
which is a formalism to describe arbitrary physical theories admitting
probabilistic processes. OPTs differ from the convex set approach
to general probabilistic theories \citep{Barrett,Barnum-2,Barnum2016}
in that they take the composition of physical processes and systems
as a primitive. Mathematically, this is based on the graphical language
of circuits \citep{Coecke-Kindergarten,Coecke-Picturalism,Selinger,Coecke2016}
and probability theory.

\subsection{The general framework\label{subsec:The-general-framework}}

OPTs describe the experiments that can be performed on a given set
of systems by a given set of physical processes. The framework is
based on a primitive notion of composition, whereby every pair of
physical systems $\mathrm{A}$ and $\mathrm{B}$ can be combined into
a composite system $\mathrm{AB}$. Physical processes can be connected
in sequence or in parallel to build circuits, in the very same way
as the corresponding devices are connected in a laboratory to build
an experiment; for instance\begin{equation}\label{eq:circuit} \begin{aligned}\Qcircuit @C=1em @R=.7em @!R { & \multiprepareC{1}{\rho} & \qw \poloFantasmaCn{\rA} & \gate{\cA} & \qw \poloFantasmaCn{\rA'} & \gate{\cA'} & \qw \poloFantasmaCn{\rA''} &\measureD{a} \\ & \pureghost{\rho} & \qw \poloFantasmaCn{\rB} & \gate{\cB} & \qw \poloFantasmaCn{\rB'} &\qw &\qw &\measureD{b} }\end{aligned}~. \end{equation}In
this example, $\mathrm{A}$, $\mathrm{A}'$, $\mathrm{A}''$, $\mathrm{B}$,
and $\mathrm{B}'$ are \emph{systems}, $\rho$ is a bipartite \emph{state},
$\mathcal{A}$, $\mathcal{A}'$ and $\mathcal{B}$ are \emph{transformations},
$a$ and $b$ are \emph{effects}. Note that inputs are on the left
and outputs on the right. 

For generic systems $\mathrm{A}$ and $\mathrm{B}$, we denote: 
\begin{itemize}
\item the set of states of system $\mathrm{A}$ by $\mathsf{St}\left(\mathrm{A}\right)$;
\item the set of effects on $\mathrm{A}$ by $\mathsf{Eff}\left(\mathrm{A}\right)$;
\item the set of transformations from $\mathrm{A}$ to $\mathrm{B}$ by
$\mathsf{Transf}\left(\mathrm{A},\mathrm{B}\right)$;
\item the sequential composition of two transformations $\mathcal{A}$ and
$\mathcal{B}$ by $\mathcal{B}\circ\mathcal{A}$ (or $\mathcal{B}\mathcal{A}$,
for short), with the input of $\mathcal{B}$ matching the output of
$\mathcal{A}$;
\item the identity transformation on system $\mathrm{A}$ by $\mathcal{I}_{\mathrm{A}}$,
represented with a plain wire $\Qcircuit @C=1em @R=.7em @!R { & \qw \poloFantasmaCn{\rA} & \qw }$;
\item the parallel composition (or tensor product) of the transformations
$\mathcal{A}$ and $\mathcal{B}$ by $\mathcal{A}\otimes\mathcal{B}$.
\end{itemize}
Among the list of valid physical systems, every OPT includes the trivial
system $\mathrm{I}$, corresponding to the degrees of freedom ignored
by theory, and to the lack of input (or output) system. States (resp.\ effects)
are transformations with the trivial system as input (resp.\ output).

A circuit with no external wires, as in Eq.~\eqref{eq:circuit},
is identified with a real number in the interval $\left[0,1\right]$,
interpreted as the probability of the joint occurrence of all the
transformations present in the circuit. We will often use the notation
$\left(a\middle|\rho\right)$ to denote the circuit \[ \left(a\middle|\rho\right)~:=\!\!\!\!\begin{aligned}\Qcircuit @C=1em @R=.7em @!R { & \prepareC{\rho} & \qw \poloFantasmaCn{\rA} &\measureD{a}}\end{aligned}~, \]and
the notation $\left(b\middle|\mathcal{C}\middle|\rho\right)$ to mean
the circuit\[\begin{aligned}\Qcircuit @C=1em @R=.7em @!R { & \prepareC{\rho} & \qw \poloFantasmaCn{\rA} &\gate{\cC} &\qw \poloFantasmaCn{\rB} &\measureD{b}}\end{aligned}~. \]

Let us clarify these concepts in quantum theory.
\begin{example}
In quantum theory, we associate a Hilbert space $\mathcal{H}^{A}$
with every system $\mathrm{A}$. States are positive semi-definite
operators $\rho$ with $\mathrm{Tr}\left[\rho\right]\leq1$. The reason
why we also consider states with trace less than 1 will be explained
in example~\ref{exa:quantum tests}. An effect is, instead, represented
by a positive semi-definite operator $E$, with $E\leq I$, where
$I$ is the identity operator. The pairing between states and effects
is given by the trace: $\left(E\middle|\rho\right)=\mathrm{Tr}\left[E\rho\right]$.
\end{example}

The fact that some circuits represent real numbers induces a notion
of sum for transformations, so that the sets $\mathsf{St}\left(\mathrm{A}\right)$,
$\mathsf{Transf}\left(\mathrm{A},\mathrm{B}\right)$, and $\mathsf{Eff}\left(\mathrm{A}\right)$
become spanning sets of\emph{ real} vector spaces. We will denote
the vector space of states as $\mathsf{St}_{\mathbb{R}}\left(\mathrm{A}\right)$
and the vector space of transformations as $\mathsf{Transf}_{\mathbb{R}}\left(\mathrm{A},\mathrm{B}\right)$.
Effects become linear functionals on $\mathsf{St}_{\mathbb{R}}\left(\mathrm{A}\right)$,
and transformations in $\mathsf{Transf}\left(\mathrm{A},\mathrm{B}\right)$
are linear transformations from $\mathsf{St}_{\mathbb{R}}\left(\mathrm{A}\right)$
to $\mathsf{St}_{\mathbb{R}}\left(\mathrm{B}\right)$.

If we restrict ourselves to linear combinations of states with non-negative
coefficients (conical combinations), we obtain a proper convex cone
\citep{Chiribella-purification}, called the cone of states $\mathsf{St}_{+}\left(\mathrm{A}\right)$.
Note that effects take non-negative values on the cone of states.
Indeed if $\xi\in\mathsf{St}_{+}\left(\mathrm{A}\right)$, then $\xi$
is a conical combination of some states $\rho_{i}$: $\xi=\sum_{i}\lambda_{i}\rho_{i}$,
where $\lambda_{i}\geq0$ for every $i$. Therefore when an effect
$a$ acts on $\xi$, we have
\[
\left(a\middle|\xi\right)=\sum_{i}\lambda_{i}\left(a\middle|\rho_{i}\right)\geq0,
\]
as $\lambda_{i}\geq0$, and $0\leq\left(a\middle|\rho_{i}\right)\leq1$,
because an effect yields a probability when applied to a state.
\begin{example}
In quantum theory, $\mathsf{St}_{\mathbb{R}}\left(\mathrm{A}\right)$
is the vector space of hermitian matrices on $\mathcal{H}^{A}$, and
$\mathsf{St}_{+}\left(\mathrm{A}\right)$ is the cone of positive
semi-definite matrices.
\end{example}

In general, an experiment in a laboratory can be non-deterministic,
i.e.\ it can result into a set of alternative transformations applied
to the input system, heralded by different outcomes, which can (at
least in principle) be accessed by an experimenter. General non-deterministic
processes are described by \emph{tests}: a test from $\mathrm{A}$
to $\mathrm{B}$ is a collection of transformations $\left\{ \mathcal{C}_{x}\right\} _{x\in X}$
from $\mathrm{A}$ to $\mathrm{B}$, where $X$ is the set of outcomes.
If $\mathrm{A}$ (resp.\ $\mathrm{B}$) is the trivial system, the
test is called a \emph{preparation-test} (resp.\ \emph{observation-test}).
If the set of outcomes $X$ contains a single element, we say that
the test is \emph{deterministic}, because only one transformation
can occur, and we can predict the outcome of the experiment. We refer
to deterministic transformations as \emph{channels}. If we sum over
all the transformations in a test we get a deterministic transformation,
viz.\ a channel: $\mathcal{C}:=\sum_{x\in X}\mathcal{C}_{x}$. This
is because the sum of all the transformations arising in a test can
be viewed as the full coarse-graining over all outcomes \citep{Chiribella-purification},
resulting in a new, deterministic, test. 
\begin{example}
\label{exa:quantum tests}In quantum theory, a channel from $\mathrm{A}_{0}$
to $\mathrm{A}_{1}$ is a CPTP map from $\mathfrak{B}\left(\mathcal{H}^{A_{0}}\right)$
to $\mathfrak{B}\left(\mathcal{H}^{A_{1}}\right)$. A test from $\mathrm{A}_{0}$
to $\mathrm{A}_{1}$ is a collection of CPTNI maps from $\mathfrak{B}\left(\mathcal{H}^{A_{0}}\right)$
to $\mathfrak{B}\left(\mathcal{H}^{A_{1}}\right)$ summing to a CPTP
map. Note that this is consistent with the fact that the sum over
all the transformations in a test yields a channel.

Deterministic states are positive semi-definite operators $\rho$
with $\mathrm{Tr}\left[\rho\right]=1$. A non-deterministic preparation-test
is a collection of positive semi-definite operators $\rho_{i}$ with
$\mathrm{Tr}\left[\rho_{i}\right]<1$ (non-deterministic states) that
sum to a deterministic state $\rho$. This is essentially a random
preparation: a state $\rho_{i}$ is prepared with a probability given
by $\mathrm{Tr}\left[\rho_{i}\right]$. This is why we consider all
positive semi-definite operators $\rho$ with $\mathrm{Tr}\left[\rho\right]\leq1$
as states.

Observation-tests are POVMs. In quantum theory there is only one deterministic
effect: the identity $I$ (more precisely it is the functional $\mathrm{Tr}\left[I\bullet\right]$).
This is not a coincidence, but it follows from the fact that quantum
theory is a causal theory (see definition~\ref{def:causality}).
\end{example}

Among all theories, \emph{causal} theories \citep{Chiribella-purification}
are particularly important: in these theories, loosely speaking, information
cannot come back from the future. They are particularly simple in
their structure, and, generally speaking, they are well understood.
Causality can also be shown to imply no-signalling in space-like separated
systems \citep{Chiribella-purification}. The formal statement of
the property of causality is as follows.
\begin{ax}[Causality \citep{Chiribella-purification}]
\label{def:causality}For every state $\rho$, take two observation-tests
$\left\{ a_{x}\right\} _{x\in X}$ and $\left\{ b_{y}\right\} _{y\in Y}$.
One has
\[
\sum_{x\in X}\left(a_{x}\middle|\rho\right)=\sum_{y\in Y}\left(b_{y}\middle|\rho\right).
\]
\end{ax}

Causality can be equivalently characterized in terms of deterministic
effects: an OPT is \emph{causal} if and only if, for every system
$\mathrm{A}$, there is a unique deterministic effect $u_{\mathrm{A}}$
\citep{Chiribella-purification}. This characterization is very practical
to work with.
\begin{example}
In quantum theory there is only one deterministic effect, the identity
operator (or the trace functional). Hence quantum theory is causal.
\end{example}

Causal theories enjoy an important property: the unique deterministic
effect for a composite system $\mathrm{AB}$ always factorizes as
the parallel composition of the deterministic effects on $\mathrm{A}$
and on $\mathrm{B}$. In symbols, $u_{\mathrm{AB}}=u_{\mathrm{A}}\otimes u_{\mathrm{B}}$.
This is because if $u_{\mathrm{A}}$ and $u_{\mathrm{B}}$ are the
deterministic effects of $\mathrm{A}$ and $\mathrm{B}$, then $u_{\mathrm{A}}\otimes u_{\mathrm{B}}$
is a deterministic effect on $\mathrm{AB}$. Since the theory is causal,
there is a unique deterministic effect on $\mathrm{AB}$, so $u_{\mathrm{A}}\otimes u_{\mathrm{B}}$
is \emph{the} deterministic effect of $\mathrm{AB}$.

Moreover, in causal theories there is a nice characterization of channels:
a transformation $\mathcal{C}\in\mathsf{Transf}\left(\mathrm{A},\mathrm{B}\right)$
is a channel if and only if \citep{Chiribella-purification}
\begin{equation}
u_{\mathrm{B}}\mathcal{C}=u_{\mathrm{A}}.\label{eq:characterization of channels}
\end{equation}
In quantum theory, since $u$ is the trace, this condition amounts
to saying that channels are trace-preserving.

Let us conclude this section by showing how the theory of quantum
supermaps fits into the OPT formalism.
\begin{example}
In the theory of quantum supermaps, every system $\mathrm{A}$ is
a pair of input and output quantum systems $\mathrm{A}=\left(A_{0},A_{1}\right)$;
deterministic states are CPTP maps, and non-deterministic ones are
CPTNI maps. The cone of states is given by all CP maps. Transformations
in this theory are supermaps \citep{Chiribella2008,Switch,Perinotti1,Perinotti2,Gour2018}.
As our results show, it is not immediate to pin down the mathematical
properties that make a generic linear supermap from $\mathrm{A}$
to $\mathrm{B}$ physical. We will analyse this issue from the OPT
perspective in the next subsection.

Now let us show that the theory of quantum supermaps is \emph{not}
causal. Suppose we want to construct a deterministic effect in this
theory. According to \citep{Circuit-architecture,Chiribella2008,Hierarchy-combs},
to this end it is enough to consider a 1-comb made of deterministic
quantum operations, which means a circuit fragment of the form\[
\begin{aligned}\Qcircuit @C=1em @R=.7em @!R {& \qw \poloFantasmaCn{\rB_0}  & \multigate{1}{\cA_0} & \qw \poloFantasmaCn{\rA_0} &\qw & & &\qw \poloFantasmaCn{\rA_1} & \multigate{1}{\cA_1} & \qw \poloFantasmaCn{\rB_1} &\qw \\ && \pureghost{\cA_0} & \qw \poloFantasmaCn{\rS}  & \qw &\qw &\qw &\qw & \ghost{\cA_1}}\end{aligned}~,
\]where both $\mathcal{A}_{0}$ and $\mathcal{A}_{1}$ are deterministic
quantum operations. Since this comb must output a probability, its
pre-processing $\mathcal{A}_{0}$ must be a deterministic bipartite
quantum state $\rho\in\mathsf{St}\left(\mathrm{A}_{0}\mathrm{S}\right)$,
and its post-processing $\mathcal{A}_{1}$ must be a deterministic
bipartite quantum effect $u\in\mathsf{Eff}\left(\mathrm{A}_{1}\mathrm{S}\right)$,
for some system $\mathrm{S}$:\[
\begin{aligned}\Qcircuit @C=1em @R=.7em @!R {&\multiprepareC{1}{\rho} & \qw \poloFantasmaCn{\rA_0} &\qw & & &\qw \poloFantasmaCn{\rA_1} & \multimeasureD{1}{u}\\ &\pureghost{\rho} & \qw \poloFantasmaCn{\rS}  & \qw &\qw &\qw &\qw & \ghost{u}}\end{aligned}~.
\]Now recall that in causal theories the deterministic effect of a bipartite
system $\mathrm{A}_{1}\mathrm{S}$ factorizes as $u_{\mathrm{A}_{1}}\otimes u_{\mathrm{S}}$,
and that $u$ is nothing but the trace (cf.\ example~\ref{exa:quantum tests}).
Then\[
\begin{aligned}\Qcircuit @C=1em @R=.7em @!R {&\multiprepareC{1}{\rho} & \qw \poloFantasmaCn{\rA_0} &\qw & & &\qw \poloFantasmaCn{\rA_1} & \multimeasureD{1}{u}\\ &\pureghost{\rho} & \qw \poloFantasmaCn{\rS}  & \qw &\qw &\qw &\qw & \ghost{u}}\end{aligned}~=\!\!\!\!\begin{aligned}\Qcircuit @C=1em @R=.7em @!R {&\multiprepareC{1}{\rho} & \qw \poloFantasmaCn{\rA_0} &\qw & & &\qw \poloFantasmaCn{\rA_1} & \measureD{\Tr}\\ &\pureghost{\rho} & \qw \poloFantasmaCn{\rS}  & \qw &\qw &\qw &\qw & \measureD{\Tr}}\end{aligned}~=\!\!\!\!\begin{aligned}\Qcircuit @C=1em @R=.7em @!R {&\prepareC{\rho'} & \qw \poloFantasmaCn{\rA_0} &\qw & & &\qw \poloFantasmaCn{\rA_1} & \measureD{\Tr}}\end{aligned}~,
\]where $\rho'=\mathrm{Tr}_{\mathrm{S}}\left[\rho\right]$. In this
way, for any choice of $\rho\in\mathfrak{D}\left(\mathcal{H}^{A_{0}S}\right)$
we obtain \emph{all} quantum states $\rho'\in\mathfrak{D}\left(\mathcal{H}^{A_{0}}\right)$.
Hence, the generic deterministic effect on system $\mathrm{A}=\left(A_{0},A_{1}\right)$
of the theory of quantum supermaps is of the form\[
u_{\rho}=\!\!\!\!\begin{aligned}\Qcircuit @C=1em @R=.7em @!R {&\prepareC{\rho} & \qw \poloFantasmaCn{\rA_0} &\qw & & &\qw \poloFantasmaCn{\rA_1} & \measureD{\Tr}}\end{aligned}~,
\]for any quantum state $\rho\in\mathfrak{D}\left(\mathcal{H}^{A_{0}}\right)$.
This means that there is a whole family of deterministic effects,
labelled by quantum states. Therefore, the theory of quantum supermaps
is \emph{not} causal, a fact that is confirmed by the presence of
signalling bipartite quantum channels \citep{Beckman}. The failure
of causality implies here that there are some deterministic effects
for a bipartite system $\mathrm{AB}=\left(A_{0},A_{1}\right)\left(B_{0},B_{1}\right)$
that do not factorize. Indeed, if we take an entangled bipartite quantum
state $\rho\in\mathfrak{D}\left(\mathcal{H}^{A_{0}B_{0}}\right)$,
the associated deterministic effect is\begin{equation}\label{eq:bipartite u}
u_{\rho}=\!\!\!\!\begin{aligned}\Qcircuit @C=1em @R=.7em @!R {&\multiprepareC{1}{\rho} & \qw \poloFantasmaCn{\rA_0} &\qw & & &\qw \poloFantasmaCn{\rA_1} & \measureD{\Tr} \\ &\pureghost{\rho} & \qw \poloFantasmaCn{\rB_0} &\qw & & &\qw \poloFantasmaCn{\rB_1} & \measureD{\Tr}}\end{aligned}~,
\end{equation}which does not factorize. This fact will play an important role in
Appendix~\ref{subsec:physical maps}, and it is ultimately the reason
why we need the CPTNI preservation condition in a complete sense.
\end{example}

\subsection{\label{subsec:physical maps}Necessary conditions for physical transformations}

In the OPT approach, however we construct a diagram, this represents
a physical object: a valid state, a valid transformation, a valid
effect. Specializing our analysis to transformations from a system
$\mathrm{A}$ to a system $\mathrm{B}$, a linear map $\mathcal{A}$
from $\mathsf{St}_{\mathbb{R}}\left(\mathrm{A}\right)$ to $\mathsf{St}_{\mathbb{R}}\left(\mathrm{B}\right)$
is a valid physical transformation only if\begin{equation} \label{eq:state} \begin{aligned}\Qcircuit @C=1em @R=.7em @!R { & \multiprepareC{1}{\rho} & \qw \poloFantasmaCn{\rA} & \gate{\cA} & \qw \poloFantasmaCn{\rB} & \qw \\ & \pureghost{\rho} & \qw \poloFantasmaCn{\rS} & \qw &\qw &\qw }\end{aligned} \end{equation}is
a valid state of system $\mathrm{BS}$, for every choice of $\rho$
and $\mathrm{S}$. Here we will derive some \emph{necessary} conditions
to guarantee this. In particular, if \eqref{eq:state} is a valid
state, for every bipartite effect $E\in\mathsf{Eff}\left(\mathrm{BS}\right)$
we have\begin{equation} \label{eq:state-inequality}
0\leq\!\!\!\!\begin{aligned}\Qcircuit @C=1em @R=.7em @!R { & \multiprepareC{1}{\rho} & \qw \poloFantasmaCn{\rA} & \gate{\cA} & \qw \poloFantasmaCn{\rB} & \multimeasureD{1}{E} \\ & \pureghost{\rho} & \qw \poloFantasmaCn{\rS} & \qw &\qw &\ghost{E} }\end{aligned} ~\leq 1,
\end{equation}because this is the probability of $E$ occurring on $\left(\mathcal{A}\otimes\mathcal{I}_{\mathrm{S}}\right)\rho$.
\begin{rem}
Condition~\eqref{eq:state-inequality} is only \emph{necessary},
but in general \emph{not sufficient} to guarantee that \eqref{eq:state}
represents a valid physical state. Indeed, the theory may have additional
restrictions on the allowed states, as it happens in the presence
of superselection rules \citep{Preskill-superselection,Fermionic2,TowardsThermo,Purity,ScandoloPhD}.
If the theory is completely unrestricted, such as quantum theory or
the theory of quantum supermaps, condition~\eqref{eq:state-inequality}
and those we derive in the following will be sufficient as well.
\end{rem}

Let us analyse the two inequalities in \eqref{eq:state-inequality}
separately. If $\left(\mathcal{A}\otimes\mathcal{I}_{\mathrm{S}}\right)\rho$
is in the cone of states of $\mathrm{BS}$, then we immediately have\[
\begin{aligned}\Qcircuit @C=1em @R=.7em @!R { & \multiprepareC{1}{\rho} & \qw \poloFantasmaCn{\rA} & \gate{\cA} & \qw \poloFantasmaCn{\rB} & \multimeasureD{1}{E} \\ & \pureghost{\rho} & \qw \poloFantasmaCn{\rS} & \qw &\qw &\ghost{E} }\end{aligned} ~\geq 0,
\]for every effect $E\in\mathsf{Eff}\left(\mathrm{BS}\right)$.
\begin{defn}
\label{def:CP-OPT}We say that a transformation $\mathcal{A}$ in
$\mathsf{Transf}_{\mathbb{R}}\left(\mathrm{A},\mathrm{B}\right)$
is \emph{completely positive} if, for every system $\mathrm{S}$ and
every element $\xi\in\mathsf{St}_{+}\left(\mathrm{AS}\right)$, we
have $\left(\mathcal{A}\otimes\mathcal{I}_{\mathrm{S}}\right)\xi\in\mathsf{St}_{+}\left(\mathrm{BS}\right)$.
\end{defn}

In words, a completely positive transformation is a linear transformation
that maps elements in the input cone of states to elements in the
output cone of states in a complete sense, i.e.\ even when there
is an ancillary system $\mathrm{S}$. This is clearly a necessary
condition for a transformation to be physical.

Note that it is equivalent to define complete positivity just on states
in $\mathsf{St}\left(\mathrm{AS}\right)$, instead of on generic elements
of $\xi\in\mathsf{St}_{+}\left(\mathrm{AS}\right)$: $\mathcal{A}$
is completely positive if and only if, for every system $\mathrm{S}$
and every state $\rho\in\mathsf{St}\left(\mathrm{AS}\right)$, we
have $\left(\mathcal{A}\otimes\mathcal{I}_{\mathrm{S}}\right)\rho\in\mathsf{St}_{+}\left(\mathrm{BS}\right)$.
To see the non-trivial implication, recall that if $\xi$ is a generic
element of $\mathsf{St}_{+}\left(\mathrm{AS}\right)$, it can be written
as a conical combination of states $\rho_{i}$ of $\mathrm{AS}$:
$\xi=\sum_{i}\lambda_{i}\rho_{i}$, with $\lambda_{i}\geq0$ for every
$i$. Then, if we know that $\left(\mathcal{A}\otimes\mathcal{I}_{\mathrm{S}}\right)\rho\in\mathsf{St}_{+}\left(\mathrm{BS}\right)$
for every $\rho\in\mathsf{St}\left(\mathrm{AS}\right)$, we have
\[
\left(\mathcal{A}\otimes\mathcal{I}_{\mathrm{S}}\right)\xi=\sum_{i}\lambda_{i}\left(\mathcal{A}\otimes\mathcal{I}_{\mathrm{S}}\right)\rho_{i}\in\mathsf{St}_{+}\left(\mathrm{BS}\right),
\]
because $\mathsf{St}_{+}\left(\mathrm{BS}\right)$ is closed under
conical combinations.
\begin{example}
In quantum theory, the cone of states is the cone of positive semi-definite
operators; therefore completely positive transformations in the sense
of definition~\ref{def:CP-OPT} are exactly CP maps.

In the theory of quantum supermaps, the cone of states is the cone
of CP maps. In this case, completely positive transformations are
CPP supermaps \citep{Chiribella2008,Gour2018}.
\end{example}

Now let us analyse the second inequality in \eqref{eq:state-inequality},
namely\begin{equation}\label{eq:inequality2}
\begin{aligned}\Qcircuit @C=1em @R=.7em @!R { & \multiprepareC{1}{\rho} & \qw \poloFantasmaCn{\rA} & \gate{\cA} & \qw \poloFantasmaCn{\rB} & \multimeasureD{1}{E} \\ & \pureghost{\rho} & \qw \poloFantasmaCn{\rS} & \qw &\qw &\ghost{E} }\end{aligned} ~\leq 1,
\end{equation}for every effect $E\in\mathsf{Eff}\left(\mathrm{BS}\right)$. Assume
$\mathcal{A}$ is completely positive. Then, demanding the validity
of inequality~\eqref{eq:inequality2} for every state $\rho\in\mathsf{St}\left(\mathrm{AS}\right)$
and every effect $E\in\mathsf{Eff}\left(\mathrm{BS}\right)$ is equivalent
to demanding its validity when $\rho$ is any \emph{deterministic}
state and $E$ any \emph{deterministic} effect. To see the non-trivial
implication, recall that if $\rho$ is non-deterministic, it arises
in a preparation-test $\left\{ \rho,\rho'\right\} $. Similarly, if
$E$ is non-deterministic, it arises in an observation-test $\left\{ E,E'\right\} $.
Clearly $\widetilde{\rho}=\rho+\rho'$ is a deterministic state, and
$\widetilde{E}=E+E'$ is a deterministic effect. Then\begin{eqnarray*}
1&\geq&\begin{aligned}\Qcircuit @C=1em @R=.7em @!R { & \multiprepareC{1}{\widetilde{\rho}} & \qw \poloFantasmaCn{\rA} & \gate{\cA} & \qw \poloFantasmaCn{\rB} & \multimeasureD{1}{\widetilde{E}} \\ & \pureghost{\widetilde{\rho}} & \qw \poloFantasmaCn{\rS} & \qw &\qw &\ghost{\widetilde{E}} }\end{aligned} \\ &=&\!\!\!\!\begin{aligned}\Qcircuit @C=1em @R=.7em @!R { & \multiprepareC{1}{\rho} & \qw \poloFantasmaCn{\rA} & \gate{\cA} & \qw \poloFantasmaCn{\rB} & \multimeasureD{1}{E} \\ & \pureghost{\rho} & \qw \poloFantasmaCn{\rS} & \qw &\qw &\ghost{E} }\end{aligned}~+\!\!\!\!\begin{aligned}\Qcircuit @C=1em @R=.7em @!R { & \multiprepareC{1}{\rho} & \qw \poloFantasmaCn{\rA} & \gate{\cA} & \qw \poloFantasmaCn{\rB} & \multimeasureD{1}{E'} \\ & \pureghost{\rho} & \qw \poloFantasmaCn{\rS} & \qw &\qw &\ghost{E'} }\end{aligned} \\ &+& \!\!\!\!\begin{aligned}\Qcircuit @C=1em @R=.7em @!R { & \multiprepareC{1}{\rho'} & \qw \poloFantasmaCn{\rA} & \gate{\cA} & \qw \poloFantasmaCn{\rB} & \multimeasureD{1}{E} \\ & \pureghost{\rho'} & \qw \poloFantasmaCn{\rS} & \qw &\qw &\ghost{E} }\end{aligned}~+\!\!\!\!\begin{aligned}\Qcircuit @C=1em @R=.7em @!R { & \multiprepareC{1}{\rho'} & \qw \poloFantasmaCn{\rA} & \gate{\cA} & \qw \poloFantasmaCn{\rB} & \multimeasureD{1}{E'} \\ & \pureghost{\rho'} & \qw \poloFantasmaCn{\rS} & \qw &\qw &\ghost{E'} }\end{aligned}~.
\end{eqnarray*}Now, each term in the right-hand side is non-negative because $\mathcal{A}$
is completely positive. It follows that each term is also less than
or equal to 1, and specifically\[
\begin{aligned}\Qcircuit @C=1em @R=.7em @!R { & \multiprepareC{1}{\rho} & \qw \poloFantasmaCn{\rA} & \gate{\cA} & \qw \poloFantasmaCn{\rB} & \multimeasureD{1}{E} \\ & \pureghost{\rho} & \qw \poloFantasmaCn{\rS} & \qw &\qw &\ghost{E} }\end{aligned} ~\leq 1.
\]We summarize these necessary requirements in the following theorem.
\begin{thm}
\label{thm:necessary condition physical}Let $\mathcal{A}\in\mathsf{Transf}_{\mathbb{R}}\left(\mathrm{A},\mathrm{B}\right)$.
Then $\mathcal{A}$ is a physical transformation only if both these
conditions are satisfied:
\begin{enumerate}
\item \label{enu:CPP}$\left(\mathcal{A}\otimes\mathcal{I}_{\mathrm{S}}\right)\rho\in\mathsf{St}_{+}\left(\mathrm{BS}\right)$
for every system $\mathrm{S}$ and every state $\rho\in\mathsf{St}\left(\mathrm{AS}\right)$;
\item \label{enu:2}\[
\begin{aligned}\Qcircuit @C=1em @R=.7em @!R { & \multiprepareC{1}{\rho} & \qw \poloFantasmaCn{\rA} & \gate{\cA} & \qw \poloFantasmaCn{\rB} & \multimeasureD{1}{u} \\ & \pureghost{\rho} & \qw \poloFantasmaCn{\rS} & \qw &\qw &\ghost{u} }\end{aligned} ~\leq 1,
\]for every system $\mathrm{S}$, every \emph{deterministic} state $\rho\in\mathsf{St}\left(\mathrm{AS}\right)$,
and every \emph{deterministic} effect $u\in\mathsf{Eff}\left(\mathrm{BS}\right)$.
\end{enumerate}
\end{thm}

Note that in particular, condition~\ref{enu:2} implies that\begin{equation}\label{eq:3 weak}
\begin{aligned}\Qcircuit @C=1em @R=.7em @!R { & \prepareC{\rho} & \qw \poloFantasmaCn{\rA} & \gate{\cA} & \qw \poloFantasmaCn{\rB} & \measureD{u}}\end{aligned} ~\leq 1,
\end{equation}for it is enough to take $\mathrm{S}$ to be the trivial system $\mathrm{I}$.
However, in general, this latter condition is \emph{weaker} than condition~\ref{enu:2},
such as in the theory of quantum supermaps. Let us analyse the role
of conditions~\ref{enu:CPP}, \ref{enu:2}, and \eqref{eq:3 weak}
in this theory.
\begin{example}
First of all, since the theory of quantum supermaps has no restrictions,
the conditions in theorem~\ref{thm:necessary condition physical}
become \emph{sufficient} as well. We have already examined condition~\ref{enu:CPP}.
Let us focus on condition~\ref{enu:2}, and unfold its meaning. In
this case, $\rho$ is actually a bipartite channel $\mathcal{N}$,
and $\mathcal{A}$ acts as a supermap $\Theta$ on half of $\mathcal{N}$.
Recalling Eq.~\eqref{eq:bipartite u}, condition~\ref{enu:2} becomes
$\mathrm{Tr}_{B_{1}S_{1}}\left[\left(\Theta^{A\rightarrow B}\otimes\mathbf{1}^{S}\right)\left[\mathcal{N}^{AS}\right]\left(\rho^{B_{0}S_{0}}\right)\right]\leq1$.
This is nothing but requiring that $\Theta$ be completely CPTNI-preserving
(cf.\ Eq.~\eqref{eq:c-CPTNI} in the main article).

In conclusion, the two conditions of theorem \ref{thm:necessary condition physical}
are exactly the two conditions we found in this article. Note that
condition~\eqref{eq:3 weak}, expressing CPTNI preservation (but
not in a complete sense), is \emph{weaker} than condition~\ref{enu:2},
as there is no way to recover condition~\ref{enu:2} from condition~\eqref{eq:3 weak}.
This is essentially because not all bipartite deterministic effects
can be reduced to single-system deterministic effects (cf.\ Eq.~\eqref{eq:bipartite u}).
Thus condition~\eqref{eq:3 weak} \emph{cannot} be used to assess
whether a candidate supermap is physical or not, so CPTNI preservation
is not enough.
\end{example}

If theorem~\ref{thm:necessary condition physical} is valid in all
physical theories, why do we not need to impose the trace non-increasing
condition in a complete sense in quantum theory? This is because the
theory is causal. Indeed in all causal theories, condition~\ref{enu:2}
becomes equivalent to condition~\eqref{eq:3 weak}.
\begin{prop}
\label{prop:TNI causal}In a causal theory with deterministic effect
$u$, one has\[
\begin{aligned}\Qcircuit @C=1em @R=.7em @!R { & \multiprepareC{1}{\rho} & \qw \poloFantasmaCn{\rA} & \gate{\cA} & \qw \poloFantasmaCn{\rB} & \multimeasureD{1}{u} \\ & \pureghost{\rho} & \qw \poloFantasmaCn{\rS} & \qw &\qw &\ghost{u} }\end{aligned} ~\leq 1,
\]for every system $\mathrm{S}$ and every deterministic state $\rho\in\mathsf{St}\left(\mathrm{AS}\right)$,
if and only if\[
\begin{aligned}\Qcircuit @C=1em @R=.7em @!R { & \prepareC{\rho} & \qw \poloFantasmaCn{\rA} & \gate{\cA} & \qw \poloFantasmaCn{\rB} & \measureD{u}}\end{aligned} ~\leq 1.
\]for every deterministic state $\rho\in\mathsf{St}\left(\mathrm{A}\right)$.
\end{prop}

\begin{proof}
We have already seen one implication (necessity), now let us focus
on the other. Assume condition~\eqref{eq:3 weak} holds. Take an
arbitrary system $\mathrm{S}$ and an arbitrary deterministic state
$\Sigma\in\mathsf{St}\left(\mathrm{AS}\right)$. Then\[
\begin{aligned}\Qcircuit @C=1em @R=.7em @!R { & \multiprepareC{1}{\Sigma} & \qw \poloFantasmaCn{\rA} & \gate{\cA} & \qw \poloFantasmaCn{\rB} & \multimeasureD{1}{u} \\ & \pureghost{\Sigma} & \qw \poloFantasmaCn{\rS} & \qw &\qw &\ghost{u} }\end{aligned} ~=\!\!\!\!\begin{aligned}\Qcircuit @C=1em @R=.7em @!R { & \multiprepareC{1}{\Sigma} & \qw \poloFantasmaCn{\rA} & \gate{\cA} & \qw \poloFantasmaCn{\rB} & \measureD{u} \\ & \pureghost{\Sigma} & \qw \poloFantasmaCn{\rS} & \qw &\qw &\measureD{u} }\end{aligned}~=\!\!\!\!\begin{aligned}\Qcircuit @C=1em @R=.7em @!R { & \prepareC{\rho} & \qw \poloFantasmaCn{\rA} & \gate{\cA} & \qw \poloFantasmaCn{\rB} & \measureD{u} }\end{aligned}~\leq 1,
\]where we have used the fact that the deterministic effect of a composite
system factorizes, and that\[
\begin{aligned}\Qcircuit @C=1em @R=.7em @!R { & \multiprepareC{1}{\Sigma} & \qw \poloFantasmaCn{\rA} & \qw \\ & \pureghost{\Sigma} & \qw \poloFantasmaCn{\rS} & \measureD{u} }\end{aligned} ~=:\!\!\!\!\begin{aligned}\Qcircuit @C=1em @R=.7em @!R { & \prepareC{\rho} & \qw \poloFantasmaCn{\rA} & \qw}\end{aligned}
\]is a deterministic state.
\end{proof}
In other words, for causal theories condition~\ref{enu:2} can be
formulated only for single systems, without the need of an ancillary
system $\mathrm{S}$. Recall that in quantum theory $u$ is the trace,
so condition~\eqref{eq:3 weak} means that $\mathcal{A}$ is trace-non-increasing.
Proposition~\ref{prop:TNI causal} is the ultimate reason why in
quantum theory it is enough to require that a CP map be TNI (on single
system) rather than \emph{completely} TNI. In conclusion, the ultimate
origin of the unexpected behaviour of the theory of quantum supermaps
is the failure of causality.

However, in \citep{Gour2018} one of the authors showed that for a
CPP map to be a superchannel, instead, it is not necessary to demand
that it be completely TPP, but it is enough that it be TPP. Why do
we not need CPTP preservation in a complete sense for superchannels?
Let us understand it using the OPT formalism.

Clearly, a superchannel $\Theta^{A\rightarrow B}$ must send channels
to channels in a complete sense: for any bipartite quantum channel
$\mathcal{N}^{AB}$, $\mathbf{1}^{R}\otimes\Theta^{A\rightarrow B}\left[\mathcal{N}^{RA}\right]=\mathcal{M}^{RB}$,
where $\mathcal{M}^{RB}$ is still a quantum channel. By Eq.~\eqref{eq:characterization of channels},
this is true if and only if
\begin{equation}
\left(\mathrm{Tr}_{R_{1}}\otimes\mathrm{Tr}_{B_{1}}\right)\circ\left(\mathbf{1}^{R}\otimes\Theta^{A\rightarrow B}\left[\mathcal{N}^{RA}\right]\right)=\mathrm{Tr}_{R_{0}}\otimes\mathrm{Tr}_{B_{0}},\label{eq:CTPP}
\end{equation}
where we have denoted the deterministic effect $u$ explicitly as
the trace. Now let us try to prove Eq.~\eqref{eq:CTPP} knowing that
$\Theta^{A\rightarrow B}$ is \emph{just} TPP. Now consider the following
channel:\begin{equation}\label{Psi'}
\begin{aligned}\Qcircuit @C=1em @R=.7em @!R {& \qw \poloFantasmaCn{\rA_0}  & \gate{\mathcal{N}'} & \qw \poloFantasmaCn{\rA_1} &\qw }\end{aligned}~:=\!\!\!\!\begin{aligned}\Qcircuit @C=1em @R=.7em @!R {&\prepareC{\rho_0}& \qw \poloFantasmaCn{\rR_0}  & \multigate{1}{\mathcal{N}} & \qw \poloFantasmaCn{\rR_1} &\measureD{\Tr}  \\ &&\qw\poloFantasmaCn{\rA_0} & \ghost{\mathcal{N}} & \qw \poloFantasmaCn{\rA_1}  & \qw }\end{aligned}~,
\end{equation}where $\rho_{0}$ is some density matrix on $R_{0}$. Since $\Theta^{A\rightarrow B}$
is TPP, we have that $\mathcal{M}'^{B}:=\Theta^{A\rightarrow B}\left[\mathcal{N}'^{A}\right]$
is still a quantum channel. In other words
\[
\mathrm{Tr}_{B_{1}}\circ\Theta^{A\rightarrow B}\left[\mathcal{N}'^{A}\right]=\mathrm{Tr}_{B_{0}}.
\]
Then, if we take a density matrix $\sigma_{0}\in\mathfrak{D}\left(\mathcal{H}^{B_{0}}\right)$,
we have
\[
\mathrm{Tr}_{B_{1}}\circ\Theta^{A\rightarrow B}\left[\mathcal{N}'^{A}\right]\left(\sigma_{0}^{B_{0}}\right)=\mathrm{Tr}_{B_{0}}\left[\sigma_{0}^{B_{0}}\right]=1.
\]
Now, recalling the definition of $\mathcal{N}'^{A}$ in Eq.~\eqref{Psi'},
we have
\begin{equation}
\mathrm{Tr}_{R_{1}}\mathrm{Tr}_{B_{1}}\left(\mathbf{1}^{R}\otimes\Theta^{A\rightarrow B}\left[\mathcal{N}^{RA}\right]\right)\left(\rho_{0}^{R_{0}}\otimes\sigma_{0}^{B_{0}}\right)=1,\label{eq:1 product}
\end{equation}
for any $\rho_{0}\in\mathfrak{D}\left(\mathcal{H}^{R_{0}}\right)$
and any $\sigma_{0}\in\mathfrak{D}\left(\mathcal{H}^{B_{0}}\right)$.
If we manage to prove that
\[
\mathrm{Tr}_{R_{1}}\mathrm{Tr}_{B_{1}}\left(\mathbf{1}^{R}\otimes\Theta^{A\rightarrow B}\left[\mathcal{N}^{RA}\right]\right)\left(\tau^{R_{0}B_{0}}\right)=1
\]
for every bipartite state $\tau^{R_{0}B_{0}}$, then the validity
of Eq.~\eqref{eq:CTPP} is shown. Now, recall that in quantum theory
every bipartite state can be written as an affine combination of product
states. Therefore $\tau^{R_{0}B_{0}}=\sum_{j}\lambda_{j}\rho_{j}^{R_{0}}\otimes\sigma_{j}^{B_{0}}$,
with $\sum_{j}\lambda_{j}=1$. Therefore
\begin{eqnarray*}
\mathrm{Tr}_{R_{1}}\mathrm{Tr}_{B_{1}}\left(\mathbf{1}^{R}\otimes\Theta^{A\rightarrow B}\left[\mathcal{N}^{RA}\right]\right)\left(\tau^{R_{0}B_{0}}\right) & = & \mathrm{Tr}_{R_{1}}\mathrm{Tr}_{B_{1}}\left(\mathbf{1}^{R}\otimes\Theta^{A\rightarrow B}\left[\mathcal{N}^{RA}\right]\right)\left(\sum_{j}\lambda_{j}\rho_{j}^{R_{0}}\otimes\sigma_{j}^{B_{0}}\right)\\
 & = & \sum_{j}\lambda_{j}\mathrm{Tr}_{R_{1}}\mathrm{Tr}_{B_{1}}\left(\mathbf{1}^{R}\otimes\Theta^{A\rightarrow B}\left[\mathcal{N}^{RA}\right]\right)\left(\rho_{j}^{R_{0}}\otimes\sigma_{j}^{B_{0}}\right)\\
 & = & \sum_{j}\lambda_{j}\\
 & = & 1,
\end{eqnarray*}
where we have used Eq.~\eqref{eq:1 product}. This proves Eq.~\eqref{eq:CTPP},
so for quantum superchannels it is indeed enough to require that they
be TPP. Note that this proof does not use any quantum feature except
causality, which allows us to characterize quantum channels as CPTP
maps, and local tomography \citep{Chiribella-purification,hardy2013},
a property that guarantees that every deterministic bipartite state
can be written as an affine combination of deterministic product states.

The same proof also shows that any attempt to adapt it to supermaps
transforming quantum channels to CPTNI maps is bound to fail: even
if $\mathrm{Tr}_{R_{1}}\mathrm{Tr}_{B_{1}}\left(\mathbf{1}^{R}\otimes\Theta^{A\rightarrow B}\left[\mathcal{N}^{RA}\right]\right)\left(\rho_{0}^{R_{0}}\otimes\sigma_{0}^{B_{0}}\right)\leq1$,
we cannot conclude that $\mathrm{Tr}_{R_{1}}\mathrm{Tr}_{B_{1}}\left(\mathbf{1}^{R}\otimes\Theta^{A\rightarrow B}\left[\mathcal{N}^{RA}\right]\right)\left(\tau^{R_{0}B_{0}}\right)\leq1$
for every bipartite state $\tau^{R_{0}B_{0}}$. The reason is that
we are only dealing with an \emph{affine} combination, possibly even
containing negative terms. This does not allow us to conclude anything
about $\sum_{j}\lambda_{j}\mathrm{Tr}_{R_{1}}\mathrm{Tr}_{B_{1}}\left(\mathbf{1}^{R}\otimes\Theta^{A\rightarrow B}\left[\mathcal{N}^{RA}\right]\right)\left(\rho_{j}^{R_{0}}\otimes\sigma_{j}^{B_{0}}\right)$.

We conclude this Appendix with an interesting remark: sometimes, even
with a non-causal theory, the weaker condition~\eqref{eq:3 weak}
is enough to characterize which completely positive transformations
are physical, in that it becomes equivalent to the stronger condition~\ref{enu:2}
in theorem~\ref{thm:necessary condition physical}. This happens
when the only deterministic states of the theory are \emph{separable}
\citep{Chiribella-purification,Chiribella-Scandolo-entanglement}:
i.e.\ they can be written as a \emph{convex} combination of product
deterministic states. In this case, suppose we know that condition~\eqref{eq:3 weak}
holds. Let us assess $\left(u\middle|\mathcal{A}\otimes\mathcal{I}_{\mathrm{S}}\middle|\Sigma\right)$,
where $\mathrm{S}$ is an arbitrary system, $\Sigma\in\mathsf{St}\left(\mathrm{AS}\right)$
is an arbitrary deterministic state, and $u\in\mathsf{Eff}\left(\mathrm{BS}\right)$
is an arbitrary deterministic effect. We have\[
\begin{aligned}\Qcircuit @C=1em @R=.7em @!R { & \multiprepareC{1}{\Sigma} & \qw \poloFantasmaCn{\rA} & \gate{\cA} & \qw \poloFantasmaCn{\rB} & \multimeasureD{1}{u} \\ & \pureghost{\Sigma} & \qw \poloFantasmaCn{\rS} & \qw &\qw &\ghost{u} }\end{aligned} ~=\sum_{j}p_j\!\!\!\!\begin{aligned}\Qcircuit @C=1em @R=.7em @!R { & \prepareC{\alpha_j} & \qw \poloFantasmaCn{\rA} & \gate{\cA} & \qw \poloFantasmaCn{\rB} & \multimeasureD{1}{u} \\ & \prepareC{\sigma_j} & \qw \poloFantasmaCn{\rS} & \qw &\qw &\ghost{u} }\end{aligned}~=:\sum_{j}p_j\!\!\!\!\begin{aligned}\Qcircuit @C=1em @R=.7em @!R { & \prepareC{\alpha_j} & \qw \poloFantasmaCn{\rA} & \gate{\cA} & \qw \poloFantasmaCn{\rB} & \measureD{u_j} }\end{aligned}~,
\]where $\left\{ p_{j}\right\} $ is a probability distribution, $\alpha_{j}$
and $\sigma_{j}$ are deterministic states, and $u_{j}$ is the deterministic
effect defined as $u_{j}:=u_{\mathrm{BS}}\left(\mathcal{I}_{\mathrm{B}}\otimes\sigma_{j,\mathrm{S}}\right)$.
Now, each term $\left(u_{j}\middle|\mathcal{A}\middle|\rho_{j}\right)\leq1$
by condition~\eqref{eq:3 weak}, so any convex combination of them
will yield a number less than or equal to 1. In this case we were
able to prove that condition~\eqref{eq:3 weak} implies condition~\ref{enu:2}
of theorem~\ref{thm:necessary condition physical}.

We can follow the same argument when, dually, all deterministic effects
are separable. This is the case of classical supermaps, where the
non-product states in the realization of bipartite deterministic effects
(Eq.~\eqref{eq:bipartite u}) are all separable. Again, let us assume
condition~\eqref{eq:3 weak} holds, and let us assess $\left(u\middle|\mathcal{A}\otimes\mathcal{I}_{\mathrm{S}}\middle|\Sigma\right)$,
where $\mathrm{S}$ is an arbitrary system, $\Sigma\in\mathsf{St}\left(\mathrm{AS}\right)$
is an arbitrary deterministic state, and $u\in\mathsf{Eff}\left(\mathrm{BS}\right)$
is an arbitrary deterministic effect, as above. One has\[
\begin{aligned}\Qcircuit @C=1em @R=.7em @!R { & \multiprepareC{1}{\Sigma} & \qw \poloFantasmaCn{\rA} & \gate{\cA} & \qw \poloFantasmaCn{\rB} & \multimeasureD{1}{u} \\ & \pureghost{\Sigma} & \qw \poloFantasmaCn{\rS} & \qw &\qw &\ghost{u} }\end{aligned} ~=\sum_{j}p_j\!\!\!\!\begin{aligned}\Qcircuit @C=1em @R=.7em @!R { & \multiprepareC{1}{\Sigma} & \qw \poloFantasmaCn{\rA} & \gate{\cA} & \qw \poloFantasmaCn{\rB} & \measureD{u_{j,\rB}} \\ & \pureghost{\Sigma} & \qw \poloFantasmaCn{\rS} & \qw &\qw &\measureD{u_{j,\rS}} }\end{aligned}~=:\sum_{j}p_j\!\!\!\!\begin{aligned}\Qcircuit @C=1em @R=.7em @!R { & \prepareC{\sigma_j} & \qw \poloFantasmaCn{\rA} & \gate{\cA} & \qw \poloFantasmaCn{\rB} & \measureD{u_{j,\rB}} }\end{aligned}~\leq 1,
\]where $\left\{ p_{j}\right\} $ is a probability distribution, $u_{j,\mathrm{B}}$
and $u_{j,\mathrm{S}}$ are deterministic effects, and $\sigma_{j}$
is a deterministic state, defined as $\sigma_{j}:=\left(\mathcal{I}_{\mathrm{A}}\otimes u_{j,\mathrm{S}}\right)\Sigma$.
The inequality follows again from condition~\eqref{eq:3 weak}.

\end{document}

%% file: myQcircuit.tex
\usepackage[matrix,frame,arrow]{xy}
\usepackage{amsmath}

\newcommand{\qw}[1][-1]{\ar @{-} [0,#1]}
\newcommand{\gate}[1]{*{\xy *+<.6em>{#1};p\save+LU;+RU **\dir{-}\restore\save+RU;+RD **\dir{-}\restore\save+RD;+LD **\dir{-}\restore\POS+LD;+LU **\dir{-}\endxy} \qw}
\newcommand{\measureD}[1]{*{\xy*+=+<.5em>{\vphantom{\rule{0em}{.1em}#1}}*\cir{r_l};p\save*!R{#1} \restore\save+UC;+UC-<.5em,0em>*!R{\hphantom{#1}}+L **\dir{-} \restore\save+DC;+DC-<.5em,0em>*!R{\hphantom{#1}}+L **\dir{-} \restore\POS+UC-<.5em,0em>*!R{\hphantom{#1}}+L;+DC-<.5em,0em>*!R{\hphantom{#1}}+L **\dir{-} \endxy} \qw}
\newcommand{\multimeasureD}[2]{*+<1em,.9em>{\hphantom{#2}}\save[0,0].[#1,0];p\save !C *{#2},p+LU+<0em,0em>;+RU+<-.8em,0em> **\dir{-}\restore\save +LD;+LU **\dir{-}\restore\save +LD;+RD-<.8em,0em> **\dir{-} \restore\save +RD+<0em,.8em>;+RU-<0em,.8em> **\dir{-} \restore \POS !UR*!UR{\cir<.9em>{r_d}};!DR*!DR{\cir<.9em>{d_l}}\restore \qw}
\newcommand{\multigate}[2]{*+<1em,.9em>{\hphantom{#2}} \qw \POS[0,0].[#1,0];p !C *{#2},p \save+LU;+RU **\dir{-}\restore\save+RU;+RD **\dir{-}\restore\save+RD;+LD **\dir{-}\restore\save+LD;+LU **\dir{-}\restore}
\newcommand{\ghost}[1]{*+<1em,.9em>{\hphantom{#1}} \qw}
\newcommand{\Qcircuit}[1][0em]{\xymatrix @*=<#1>} % @*[o]
\newcommand{\pureghost}[1]{*+<1em,.9em>{\hphantom{#1}}}
\newcommand{\multiprepareC}[2]{*+<1em,.9em>{\hphantom{#2}}\save[0,0].[#1,0];p\save !C
  *{#2},p+RU+<0em,0em>;+LU+<+.8em,0em> **\dir{-}\restore\save +RD;+RU **\dir{-}\restore\save
  +RD;+LD+<.8em,0em> **\dir{-} \restore\save +LD+<0em,.8em>;+LU-<0em,.8em> **\dir{-} \restore \POS
  !UL*!UL{\cir<.9em>{u_r}};!DL*!DL{\cir<.9em>{l_u}}\restore}
\newcommand{\prepareC}[1]{*{\xy*+=+<.5em>{\vphantom{#1\rule{0em}{.1em}}}*\cir{l^r};p\save*!L{#1} \restore\save+UC;+UC+<.5em,0em>*!L{\hphantom{#1}}+R **\dir{-} \restore\save+DC;+DC+<.5em,0em>*!L{\hphantom{#1}}+R **\dir{-} \restore\POS+UC+<.5em,0em>*!L{\hphantom{#1}}+R;+DC+<.5em,0em>*!L{\hphantom{#1}}+R **\dir{-} \endxy}}
\newcommand{\poloFantasmaCn}[1]{{{}^{#1}_{\phantom{#1}}}}